\documentclass[11pt]{article}
\usepackage{latexsym}
\usepackage{amssymb, amsmath, xspace, lscape,  latexsym}
\usepackage{amsthm}

\newcommand{\bea}{\begin{eqnarray}}
\newcommand{\eea}{\end{eqnarray}}
\def\beq#1#2\eeq{
        \begin{equation}
        \label{#1}
            #2
        \end{equation}}

\newcommand{\al}{\alpha}

\renewcommand{\tilde}{\widetilde}

\def\btheor#1\etheor{
        \begin{theor}
            #1
        \end{theor}
    }

    \def\bsled#1\esled{
        \begin{sled}
            #1
        \end{sled}   }

\newtheorem*{rem1}{Remark 1}

\newtheorem*{rem3}{Remark 3}

\newtheorem*{rem11}{Remark 11}
\newtheorem*{rem12}{Remark 12}
\newtheorem*{rem13}{Remark 13}

\newtheorem{theorem}{Theorem}

\newtheorem{lemma}{Lemma}

\newtheorem{cor}{Corollary}

\def\hm#1{#1\nobreak\discretionary{}{\hbox{\m@th$#1$}}{}}
\def\mi#1{\discretionary{\hbox{\m@th$#1$}}{\hbox{\m@th$#1$}}{}}

\vspace{4ex}
\begin{document}
\title{\bf Single-User MIMO System, Painlev\'e Transcendents and Double Scaling}
\author{Hongmei Chen, Min Chen\thanks{chen121386@163.com}, Gordon Blower\thanks{Corresponding author g.blower@lancaster.ac.uk},
Yang Chen\thanks{yayangchen@umac.mo and yangbrookchen@yahoo.co.uk}\\
        Faculty of Science and Technology, Department of Mathematics,\\
        University of Macau.\\
        Department of Mathematics, Lancaster University.}
\date{\today}
\maketitle
\begin{abstract}
In this paper we study a particular Painlev\'e V (denoted ${\rm P_{V}}$)
that arises from Multi-Input-Multi-Output (MIMO) wireless communication systems.
Such a $P_V$ appears through its intimate relation with
the Hankel determinant that describes the moment generating function (MGF) of the Shannon capacity.
This originates through the multiplication of the Laguerre weight or the
Gamma density $x^{\alpha} {\rm e}^{-x},\;x> 0,$ for $\alpha>-1$ by $(1+x/t)^{\lambda}$ with $t>0$ a scaling parameter.
Here the $\lambda$ parameter ``generates" the Shannon capacity; see Yang Chen and Matthew McKay, IEEE Trans. IT, 58 (2012) 4594--4634.
It was found that the MGF has an integral representation as a functional of $y(t)$ and $y'(t)$, where $y(t)$ satisfies
the ``classical form" of $P_V$. In this paper, we consider the situation where $n,$ the
number of transmit antennas, (or the size of the random matrix), tends to infinity, and
the signal-to-noise ratio (SNR) $P$ tends to infinity, such that $s={4n^{2}}/{P}$ is finite. Under such double scaling
the MGF, effectively an infinite determinant, has an integral representation in terms of a ``lesser" $P_{III}$.
We also consider the situations where $\alpha=k+1/2,\;\;k\in \mathbb{N},$ and $\alpha\in\{0,1,2,\dots\}$ $\lambda\in\{1,2,\dots\},$
linking the relevant quantity to a solution of the two dimensional sine-Gordon equation in radial coordinates and a certain discrete Painlev\'e-II.
\\
From the large $n$ asymptotic of the orthogonal polynomials, that appears naturally, we obtain the double
scaled MGF for small and large $s$, together with the constant term in the large $s$ expansion.
With the aid of these, we derive a number of cumulants and find that the capacity distribution function is non-Gaussian.

\end{abstract}
\noindent

\setcounter{equation}{0}
\section{Introduction.}
Multiple-Input-Multiple-Output(MIMO) wireless communication systems play
an increasingly important role in the wireless communications
development and research. The ability to increase the capacity (effectively)
 without the restrictions of power and bandwidth channels,
due to the discoveries of Telatar \cite{Telatar1999}, Foschini and Gans \cite{FoschiniGans1998},
 has spurred a tremendous amount of research.
The fundamental descriptor, the outage capacity of the information theoretic aspect of MIMO systems, is hard to
characterize, as this requires knowledge of the entire distribution of the channel mutual information.
Gaussian approximations were usually made by computing a few moments of the distribution,
see \cite{HKL2008, MC 2005, SRS2003, TV2005}.
In \cite{KMMC2011} the Coulomb fluid approximation was employed to characterize the large
 deviations behavior of the  MIMO channel capacity. With the aid of a particular Painlev\'e $P_V$ equation, Chen and McKay \cite{YangMcky2012} obtain the mean, the variance, and any
desired number of high order cumulants. Furthermore, a proposal was made in \cite{YangMcky2012}
on a double scaling scheme where $n$ tends to infinity, and the signal-noise ratio $P$ (SNR) tends to infinity, such that the parameter $s=4n^{2}/P$ is finite. One could think of this as the thermodynamic limit of such random matrix ensembles.
\par
In this paper, we carry through the double scaling and obtain the resulting double scaled moment generating function (MGF), described effectively
by an infinite determinant. It is found that the
logarithmic derivative of the double scaled Hankel determinant 
satisfies a Jimbo--Miwa--Okamoto $\sigma$-form Painlev\'e equation (essentially a $P_{III}$).  Further the
MGF has an integral representation in terms of a particular $P_{V}$ or a $P_{III},$ both in the classical form.
The asymptotic expansions of the MGF for small $s$ and large $s$ are obtained. Furthermore, with certain restrictions on the parameters, closed form expressions of the MGF are obtained.
\par
This paper is organized as follows. In Section 2, we briefly recall the single-user model. The logarithmic
derivative of the determinant of a particular $n$ dimensional Hankel matrix satisfies the Jimbo--Miwa--Okamoto $\sigma$-form
of the Painlev\'e equation, and the Hankel determinant has an integral representation in terms of a $P_{V},$
obtained by Chen and McKay \cite{YangMcky2012}. In this situation, the parameters in the $P_V$, either the $\sigma$ or
the classical form, depend on $n$.
\par
In Section 3, we show that the double scaled Hankel determinant has an integral representation
in terms of a particular $P_V$ or $P_{III}$. In this situation the parameters that occur
in the Painlev\'e equations are independent of $n$.
The formal power series expansion for small $s$ and large $s$ of the MGF are worked out.
For large $s,$ the hard-to-find constant term in asymptotic expansion is obtained
through the connection with the polynomials orthogonal with respect to the deformed Laguerre weight,
$x^{\al}{\rm e}^{-x}(t+x)^{\lambda},$ for $\alpha>-1,\; t>0,\;x\geq 0.$
\par
In Section 4, with Normand's formulas \cite{N2004}, the asymptotic expansions of the Hankel
determinant for finite $n$ and the double-scaled situation are also found. Here, with chosen special values of the parameters,
the double-scaled Hankel determinant degenerate to the asymptotic expansions obtained by Tracy and Widom \cite{TW1994}.
\par
In Section 5, we take
$\alpha=k+\frac{1}{2},$ $k\in \mathbb{N}$ and $\lambda\in\mathbb{R}$. Here one finds rational solutions in terms of
the Umemura polynomials and obtains closed form expression of the scaled MGF.
Furthermore, if $\alpha=0,$  then the $P_V$ is equivalent to the $2$-dimensional sine-Gordon equation in the
radial variable. If $\lambda\in\mathbb{N}$ this degenerates to a discrete Painlev\'e II.
Taking $\alpha=1$, and $\lambda\in\mathbb{N}$, and combining with the
Backlund transformation of the seed solution, expressed in terms of the modified Bessel-functions $I_1(z)$ and $I_0(z),$
we obtain solution of $P_{V}$ with $\alpha\in \mathbb{N},$ $\lambda\in \mathbb{N}.$
In these results, $\lambda\in \mathbb{N}$ labels the discrete $P_{II}$, a second order non-linear difference equation,
and  $\alpha\in \mathbb{N}$ denotes the number of ``applications" taken in the B\"{a}cklund transformation.
\par
In Section 6, we obtained the asymptotic behavior of the scaled MGF, and note that
the distribution of the outage capacity is no longer Gaussian.
\par
In Section 7, we display the reproducing or Christoffel--Darboux kernel constructed out of the polynomials orthogonal with respect to the
deformed Laguerre weight $x^{\alpha}e^{-x}(t+x)^{\lambda}$. We note that the
polynomials satisfy ladder operator relations, as in \cite{YangMcky2012}. With the double scaling
described above namely, $t\rightarrow 0,$ $n\rightarrow \infty,$  such that $s=4nt$ is finite,
(or equivalently $s=4n^2/P$ is finite), where
the ``coordinates" $x,$ $y$ are re-scaled to $x=\frac{X}{4n},$ $y=\frac{Y}{4n},$ it is found that scaled kernel
can be expressed in terms of the solutions of a second order ordinary differential equation (ODE). If $s=0,$
then the second order ODE becomes the Bessel differential equation. If $s>0$, one finds a Heun-like second order ODE.

\section{Single-User System Model}
In this section we give a brief description of the random matrix model which arises from a particular
wireless communication system.
See \cite{BasorChen} and \cite{YangMcky2012} and the references therein for in-depth discussion of such models. A point to point MIMO communication system has $n_{t}$ transmit and $n_{r}$ receive antennas.
The MIMO channel is usually described by
\bea\label{a1}
y=\textbf{H}x+\textbf{n},
\eea
where $\textbf{H} \in \mathbb{C}^{n_{r}\times n_{t}}$ represents the channel matrix,
and the $(j,k)$ entry $h_{jk}$ of $\textbf{H}$ denotes the wireless fading coefficients between
the $j{\rm th}$ transmit and the $k{\rm th}$ receive antenna. Here $x\in\mathbb{C}^{n_{t}}$ and $y\in \mathbb{C}^{n_{r}}$
are the input and output signal vector, and $\textbf{n}_{n_{r}\times {\rm 1}} \in \mathbb{C}^{n_{r}}$ is a complex
Gaussian distributed noise vector. Under assumption of no correlations between the input and output antennas,
$\textbf{H}$ is modeled with a complex Gaussian distribution with independently and identically distributed entries
$h_{jk}$, so$\textbf{H} \sim \mathcal{CN}(0,{\rm I_{n_{r}}})$. The input signal vector $x$ are chosen to meet a power constraint,
$$
{\rm tr}\textbf{Q}_x\leq P,\;\;(P>0),
$$
where $\textbf{Q}_x=xx*$ is covariance of a zero-mean circularly symmetric complex Gaussian distribution.
The common assumption one makes is that
$$
\textbf{Q}_x=\frac{P}{n_t}\textbf{I}_{n_t},
$$
which one interprets as sending independent complex Gaussian signals from each transmit antenna, each with power $P/n_t.$ The capacity of a communication link, due to Shannon, is determined by the ``mutual information" between the transmitter and receiver,
with the constraint given above, the mutual information reads,
$$
I(x;y,\textbf{H})=\ln\det(\textbf{I}_{n_t}+\textbf{HH}^{\dag}/t),\;\;\;t:=n_{t}/P.
$$
To proceed further, recall the notations in \cite{YangMcky2012},
$$
m={\rm max}\{n_{r},n_{t}\}, \quad n={\rm min}\{n_{r},n_{t}\}, \quad \alpha=m-n.
$$
Let $\textbf{W}$ be the complex Wishart matrix
\bea\label{a4}
\textbf{W}=\left\{
\begin{array}{cc}
{\textbf{HH}^{\dag}}, &n_{r} < n_{t},\\
{\textbf{H}^{\dag}{\textbf{H}}}, &n_{r} \geq n_{t},
\end{array}
\quad {\rm and} \quad n={\rm dimension}\{\textbf {W}\}.
\right.
\eea
If ${\{x_j}\}_{j=1}^{n}$ are the positive eigenvalues of $\textbf{W},$
then the joint probability density, up to a positive constant, reads
\bea\label{a3}
p(x_1,x_2,\dots ,x_n)dx_1\dots dx_n\propto \prod_{1\leq j<k\leq n}(x_k-x_j)^2\prod_{l=1} x_l^{\alpha}{\rm e}^{-x_l}dx_l,
\eea
where $x_j\in (0,\infty),\;j=1,\dots ,n.$
The outage probability, see for example \cite{ZhengTse2003}, describes an outage event, where communication at rate $R$
becomes impossible, reads
$$
{\rm P_{{\rm out}}(R)}={\rm Prob}(I(x; y, {\textbf{H}})<R).
$$
\par
The moment generating function (MGF) of outage capacity can be expressed as
\bea\label{a6}
\mathcal{M}(\lambda)&=&E_{\textbf{H}}\left({\rm e}^{\lambda I(x; y, {\textbf{H}})}\right) \nonumber\\
&=&\frac{t^{-n\lambda}}{Z_n}\frac{1}{n!}\int_{\mathbb{R}_{+}^{n}}\prod_{1\leq j<k\leq n}(x_k-x_j)^2\prod_{l=1}^{n}(t+x_{l})^{\lambda}x_{l}^{\alpha}e^{-x_{l}}dx_l
\eea
where $Z_n$ is a constant, chosen so that $\mathcal{M}(0)=1$. The cumulant generating function (CGF) has a power series expansion around $\lambda =0$,
\bea\label{a3}
{\rm ln}{\mathcal{M}(\lambda)}=\sum_{\ell=1}^{\infty}\kappa_{\ell}\frac{\lambda^{\ell}}{\ell!},
\eea
where  $\kappa_{\ell}$ are the cumulants. For convenience, we look at an object allied to ${\cal M}(\lambda);$
\begin{align}\label{Q1}
\mathbf{M}_{n}(t; \alpha, \lambda):&=\mathcal{M}(\lambda)t^{n\lambda}\nonumber\\
&=\frac{\frac{1}{n!}\int_{{\mathbb{R}}_{+}^{n}}\prod_{1\leq j<k\leq n}(x_{j}-x_{k})^2\prod_{\ell=1}^{n}\left(t+x_{\ell}\right)^{\lambda}x_{\ell}^{\alpha}
e^{-x_{\ell}}dx_{\ell}}{\frac{1}{n!}\int_{{\mathbb{R}}_{+}^{n}}\prod_{1\leq j<k\leq n}(x_{j}-x_{k})^2\prod_{\ell=1}^{n}x_{\ell}^{\alpha+\lambda}
e^{-x_{\ell}}dx_{\ell}}=\frac{D_{n}(t; \alpha, \lambda)}{D_{n}(0; \alpha, \lambda)},
\end{align}
where
\begin{equation}\label{Q2}
D_{n}(t; \alpha, \lambda)=\frac{1}{n!}\int_{{\mathbb{R}}_{+}^{n}}\prod_{j<k}(x_{j}-x_{k})^2\prod_{\ell=1}^{n}
\left(t+x_{\ell}\right)^{\lambda}x_{\ell}^{\alpha}
e^{-x_{\ell}}dx_{\ell},
\end{equation}
is the Hankel determinant, generated from the deformed Laguerre weight on $(0, \infty )$, defined by,
\begin{equation}\label{AA1}
w(x; t, \alpha,\lambda):=(x+t)^{\lambda}x^{\alpha}e^{-x},\;\;t>0,\;\;\alpha>-1,\; x\geq 0.
\end{equation}
Note that
$D_{n}(0; \alpha, \lambda)$ has a closed form expression,
\begin{equation}\label{Q15A}
D_{n}(0; \alpha, \lambda)=\frac{G(n+1)G(n+\lambda+\alpha+1)}{G(\lambda+\alpha+1)},
\end{equation}
with $G(z)$ is the Barnes $G$-function, and satisfies the functional equation $G(z+1)=\Gamma(z)G(z)$, $z\in \mathbb{C}\cup\{\infty\},$ the initial condition being $G(1)=1.$ See Voros \cite{Voros} for a description of $G(z).$
\par
The Painlev\'e equations that are involved in this paper are $P_{III}(\alpha, \beta, \gamma, \delta)$
\begin{equation}w''={\frac{(w')^2}{w}}-{\frac{w'}{z}}+{\frac{\alpha w^2+\beta}{z}}+\gamma w^3+{\frac{\delta}{w}};\end{equation}
and $P_{V}(\alpha, \beta, \gamma, \delta )$ given by
\begin{equation}w''={\frac{3w-1}{2w(w-1)}}(w')^2 -{\frac{w'}{z}}+{\frac{(w-1)^2}{z^2}}\Bigl( \alpha w+{\frac{\beta}{w}}\Bigr) +{\frac{\gamma w}{z}}+{\frac{\delta w(w+1)}{w-1}}.\end{equation} 

In subsequent results, these appear in  $\sigma$ form, as in the isomonodromic deformation theory of Jimbo and Miwa \cite{JimboMiwa1981}. We recall Theorem 1 in \cite{YangMcky2012}; that the log-derivative of the
Hankel determinant satisfies a particular $\sigma$-form of a Painlev\'e. 
\begin{theorem}
The logarithmic derivative of the Hankel determinant in (\ref{Q2}), associated with the deformed Laguerre weight (\ref{AA1}),
\begin{align}\label{Q3}
H_n(t):=t\frac{d}{dt}\ln{\frac{D_{n}(t; \alpha, \lambda)}{D_{n}(0; \alpha, \lambda)}},
\end{align}
satisfies
\begin{equation}\label{Q4}
(tH_{n}''(t))^{2}=\left[tH_{n}'-H_{n}+n\lambda+(2n+\alpha+\lambda)H_{n}'\right]^{2}
-4\left(tH_{n}'-H_{n}+\delta_{n}\right)\left[(H_{n}')^{2}+\lambda{H_{n}'}\right],
\end{equation}
with $\delta_{n}:=n(n+\alpha+\lambda).$ If $\sigma(t)=H_{n}(t)-n\lambda,$ then $\sigma(t)$ satisfies a version of the Jimbo--Miwa--Okamoto \cite{JimboMiwa1981, Okamoto1981} $\sigma$-form of $P_V,$
\begin{equation*}
(t\sigma'')^{2}=(\sigma-t\sigma'+2\sigma'^{2}+(v_{0}+v_{1}+v_{2}+v_{3})\sigma')^{2}-4(v_{0}+\sigma')(v_{1}+\sigma')(v_{2}+\sigma')(v_{3}+\sigma'),
\end{equation*}
with parameters
$$
v_{0}=0,\quad v_{1}=-n,\quad v_{2}=-\alpha-n,\quad v_{3}=\lambda.
$$
\end{theorem}
We list below, for reference, two coupled Riccati equations, (2.37) and (2.38), \cite{YM2010}),
\begin{equation}\label{Q5}
t\:r_{n}'(t)=\frac{r_{n}^{2}-\lambda{r_{n}}}{R_{n}}-\frac{R_{n}}{1-R_{n}}\left[r_{n}
\left(2n+\alpha+\lambda\right)+n\left(n+\alpha\right)+\frac{r_{n}^{2}-\lambda{r_{n}}}{R_{n}}\right],
\end{equation}
\begin{equation}\label{Q6}
t\:R_{n}'(t)=2r_{n}-\lambda+\left[2n+\alpha+\lambda+\left(1-R_{n}\right)t\right]R_{n}.
\end{equation}
where $r_{n}(t)$ and $R_{n}(t)$ are expressed in terms of the monic polynomials $P_{n}(x;t)$ are orthogonal with respect
to (\ref{AA1}) and can be found in \cite{YangMcky2012}:
\bea\label{A06}
R_{n}(t):=\frac{\lambda}{h_{n}(t)}\int_{0}^{\infty}\frac{(P_{n}(x))^{2}}{x+t}w(x; t, \alpha,\lambda)dx,
\eea
\bea\label{A07}
r_{n}(t):=\frac{\lambda}{h_{n-1}(t)}\int_{0}^{\infty}\frac{P_{n}(x)P_{n-1}(x)}{x+t}w(x; t, \alpha,\lambda)dx.
\eea
\\
Here $h_{n}(t)$ is the square of the $L^{2}$ norm:
\begin{equation*}
h_{n}(t):=\int_{0}^{\infty} (P_n(x))^2\;w(x;t,\al,\lambda)\;dx.
\end{equation*}
To ease the notations we do not indicate here that the orthogonal polynomials also depend upon $t,\alpha$ and $\lambda.$ Letting
\begin{equation}\label{Q7}
R_{n}(t)=1+\frac{1}{y(t)-1},\quad  {\rm or}\quad  y(t)=1+\frac{1}{R_{n}(t)-1},
\end{equation}
followed by substituting  (\ref{Q6}) and (\ref{Q7}) into (\ref{Q5}), we see that $y(t)$ satisfies
the following \par
\noindent $P_{V}(\frac{\alpha^{2}}{2},-\frac{\lambda^{2}}{2},2n+1+\alpha+\lambda,-\frac{1}{2} ),$ namely,
\begin{equation}\label{Q8}
y''(t)=\frac{3y-1}{2y(y-1)}\left(y'(t)\right)^{2}-\frac{y'(t)}{t}+\frac{(y-1)^{2}}{2t^{2}}
\left(\alpha^{2}y-\frac{\lambda^{2}}{y}\right)+\frac{(2n+1+\alpha+\lambda)y}{t}-\frac{y(y+1)}{2(y-1)}.
\end{equation}
\par
One can make a change of variable $t=e^z$, and obtain the modified $\tilde P_V$ equation, with the property that all local solutions can analytically continued to single-valued meromorphic solutions in the complex plane by Theorem 5.1 of \cite{GLS2002}. Combining ($(235),$ $(242)$, \cite{YangMcky2012}), (\ref{Q6}) and (\ref{Q7}), we see that
$H_{n}(t)$ 
 can be expressed in terms $y(t)$ and $y'(t)$ as,
\begin{equation}\label{Q9}
H_{n}(t)=\frac{t^2(y')^{2}}{4y(y-1)^{2}}+\frac{nty}{y-1}-\frac{\left(\alpha{y}+\lambda\right)^{2}}{4y}
-\frac{t^{2}y}{4(y-1)^{2}}+\frac{t(\alpha{y}+\lambda)}{2(y-1)}.
\end{equation}
With the aid of (\ref{Q1}), (\ref{Q3}) and (\ref{Q9}),
the logarithm of the MGF denoted by $\mathbf{M}_{n}(t; \alpha, \lambda)$ has the integral representation
\begin{align}\label{Q16}
\ln{\mathbf{M}}_{n}&(t; \alpha, \lambda)\nonumber\\
&=\int_{0}^{t}\left(\frac{x^2(y'(x))^{2}}{4y(x)(y(x)-1)^{2}}+\frac{nx\;y(x)}{y(x)-1}-\frac{\left(\alpha{y(x)}+\lambda\right)^{2}}{4y(x)}
-\frac{x^{2}y(x)}{4(y(x)-1)^{2}}+\frac{x(\alpha{y(x)}+\lambda)}{2(y(x)-1)}\right)\frac{dx}{x}.
\end{align}

\section{Double Scaling}
In this section, we investigate the behavior of $\mathbf{M}_{n}(t; \alpha, \lambda)$ by sending $n \rightarrow \infty$ and
$t\rightarrow 0,$ such that $s:=4nt=4n^2/P$ finite, such a double scaling scheme had been pointed out by Chen and McKay \cite{YangMcky2012}.
\\
For convenience, we first introduce two notations,
\begin{equation}\label{Q11}
g(s)=\lim_{n\rightarrow\infty} y\left(\frac{s}{4n}\right),
\end{equation}
\begin{equation}\label{Q10}
\mathcal{H}(s)=\lim_{n\rightarrow\infty} H_{n}\left(\frac{s}{4n}\right).
\end{equation}
The existence of these limits may be inferred from the finite $n$-version of Painlev\'e equations.
\begin{theorem}
Let $s=4nt,$ where $n\rightarrow\infty,$ $t\rightarrow 0,$ such that $s$ is finite. Then
$g(s)$ satisfies a particular Painlev\'{e} differential equation,
\begin{equation}\label{Q13}
\frac{d^2g}{ds^2}=\frac{3g-1}{2g(g-1)}\left(\frac{dg}{ds}\right)^2-\frac{1}{s}\frac{dg}{ds}+\frac{(g-1)^2}{2s^2}
\left(\alpha^2g-\frac{\lambda^2}{g}\right)+\frac{g}{2s},
\end{equation}
known as $P_{V}\left(\frac{\alpha^{2}}{2},-\frac{\lambda^{2}}{2},\frac{1}{2},0\right),$ subject to the initial
conditions $g(0)=-\frac{\lambda}{\alpha},\: g'(0)=\frac{\lambda}{2\alpha\left((\lambda+\alpha)^{2}-1\right)},$ and $\alpha >0,$ $\alpha+\lambda \notin \mathbb{Z}.$
\par
Moreover, $\mathcal{H}(s)$ satisfies
\begin{equation}\label{Q14}
\left(s\mathcal{H}''\right)^{2}=\mathcal{H}'\left(4\mathcal{H}'-1\right)\left(\mathcal{H}
-s\mathcal{H}'\right)+\frac{1}{16}\left[4\left(\alpha+\lambda\right)\mathcal{H}'-\lambda\right]^{2},
\end{equation}
with the initial conditions $\mathcal{H}(0)=0,$ $\mathcal{H}'(0)=\frac{\lambda}{4(\alpha+\lambda)}.$ Furthermore,
\begin{equation}\label{Q15}
\mathcal{H}(s)=\frac{(sg'(s))^{2}}{4g(s)(g(s)-1)^{2}}-\frac{\left(\lambda+\alpha{g(s)}\right)^{2}}{4g(s)}+\frac{sg(s)}{4(g(s)-1)}.
\end{equation}
\end{theorem}
\begin{proof}
With the aid of $t=s/4n,$ the equations 
follow by  taking the limit as $n\rightarrow \infty$ of (\ref{Q8}), (\ref{Q4}) and (\ref{Q9}).
\begin{rem1}
Letting $\mathcal{H}(s)=\sigma(4s)$ and combining with equation (\ref{Q14}), then $\sigma(4s)$ satisfies a $\sigma$-form ((3.13), \cite{Jimbo1982}) of $P_{III},$
$$
(s\sigma'')^2=4\sigma'(\sigma'-1)(\sigma-t\sigma')+\left(\theta_{\infty}\sigma'-\frac{\theta_{0}+\theta_{\infty}}{2}\right)^{2},
$$
with parameters $\theta_{0}=\lambda-\alpha,$ $\theta_{\infty}=\lambda+\alpha.$
\\
If we denote
\begin{equation}\label{Q12}
\mathbf{M}(s; \alpha, \lambda):=\lim_{n\rightarrow\infty}\mathbf{M}_{n}\left(\frac{s}{4n}; \alpha, \lambda\right)=\lim_{n\rightarrow\infty}{\frac{D_{n}(\frac{s}{4n}; \alpha, \lambda)}{D_{n}(0; \alpha, \lambda)}},
\end{equation}
with the initial condition $\mathbf{M}(0; \alpha, \lambda)=1$, then its logarithm has an integral representation in terms of
$P_{V}\left(\frac{\alpha^{2}}{2},-\frac{\lambda^{2}}{2},\frac{1}{2},0\right).$
With the aid of (\ref{Q1}), (\ref{Q3}) and (\ref{Q15}), one finds,
\begin{equation}\label{Q17}
\ln\mathbf{M}(s; \alpha, \lambda)=\int_{0}^{s}\left(\frac{(xg'(x))^{2}}{4g(x)(g(x)-1)^{2}}-\frac{\left(\lambda+\alpha{g(x)}\right)^{2}}{4g(x)}
+\frac{xg(x)}{4(g(x)-1)}\right)\frac{dx}{x}.
\end{equation}
\end{rem1}

\end{proof}

\subsection{Small $s$ and large $s$ behavior of the $\mathbf{M}(s; \alpha, \lambda)$.}
In this subsection we obtain formal expansion in small $s$ and large $s$ of $\mathbf{M}(s; \alpha, \lambda)$.
\\
For small and positive $s,$ let the expansion of
$P_{V}\left(\frac{\alpha^{2}}{2},-\frac{\lambda^{2}}{2},\frac{1}{2},0\right)$ be
of the form $\sum_{j=0}^{\infty}a_{j}s^{j},$ with $g(0)=-\frac{\lambda}{\alpha},$  with $g'(0)=\frac{\lambda}{2\alpha((\lambda+\alpha)^{2}-1)}.$ Substituting these into (\ref{Q13})
 we find,
\begin{equation}\label{Q22}
g(s)=-\frac{\lambda}{\alpha}+\frac{\lambda}{2\alpha((\lambda+\alpha)^{2}-1)}s
+\frac{\lambda[\lambda^{3}+\alpha{\lambda^{2}}-5\lambda(\alpha^{2}-1)+3\alpha(\alpha^{2}-1)]}
{8\alpha(\lambda+\alpha)[(\lambda+\alpha)^{2}-1]^{2}[(\lambda+\alpha)^{2}-4]}s^{2}+\mathcal{O}(s^{3}),
\end{equation}
where $\lambda+\alpha \notin \mathbb{Z}.$

In Theorem 37.5 of \cite{GLS2002}, the authors consider algebraic branch points of solutions of $P_V$ and conclude that under certain circumstances the branching order is two. Hence, for large and positive $s,$ substitute the series expansion $\sum_{j=1}^{\infty}b_{j}s^{-\frac{j}{2}}$ into (\ref{Q13}).
A straightforward computation gives,
\begin{equation}\label{Q19}
g(s)=-\lambda{s^{-\frac{1}{2}}}-\lambda^{2}s^{-1}-\frac{\lambda(8\lambda^{2}-
4\alpha^{2}+1)}{8}s^{-\frac{3}{2}}-\frac{\lambda^{2}(2\lambda^{2}-4\alpha^{2}+1)}{2}s^{-2}+\mathcal{O}\left(s^{-\frac{5}{2}}\right).
\end{equation}
\par
The theorem below gives $\mathbf{M}(s; \alpha, \lambda)$ for small $s$ and large $s.$
\begin{theorem}
The asymptotic expansions of $\mathbf{M}(s; \alpha, \lambda)$ for small and large $s,$ are:\\
\begin{enumerate}\item For small $s,$ with $\alpha+\lambda \notin \mathbb{Z}$,
\begin{align}\label{Q20}
\ln\mathbf{M}(s; \alpha, \lambda)=&\frac{\lambda}{4(\alpha+\lambda)}s-\frac{\alpha{\lambda}}{32(\alpha+\lambda)^{2}[(\alpha+\lambda)^{2}-1]}s^{2}\nonumber\\
&+\frac{\alpha{\lambda}(\alpha-\lambda)}{96(\alpha+\lambda)^{3}[(\alpha+\lambda)^{2}-1][(\alpha+\lambda)^{2}-4]}s^{3}+\mathcal{O}(s^{4}).
\end{align}
\item For large $s,$
\begin{align}\label{Q21}
\ln\mathbf{M}(s; \alpha, \lambda)=&c(\alpha, \lambda)+\lambda{s^{\frac{1}{2}}}-\frac{\lambda(2\alpha+\lambda)}{4}{\ln}s-\frac{\lambda(4\alpha^{2}-1)}{8}s^{-\frac{1}{2}}-\frac{\lambda^{2}(4\alpha^{2}-1)}{16}s^{-1}\nonumber\\
&+\frac{\lambda(4\alpha^{2}-1)(4\alpha^{2}-9-16\lambda^{2})}{384}s^{-\frac{3}{2}}+\frac{\lambda^{2}(4\alpha^{2}-1)(4\alpha^{2}-4\lambda^{2}-9)}{128}s^{-2}\nonumber\\
&{}\quad{}-\frac{\lambda(4\alpha^{2}-1)[225+16\alpha^{4}+720\lambda^{2}+128\lambda^{4}-8\alpha^{2}(17+40\lambda^{2})]}{5120}s^{-\frac{5}{2}}+\mathcal{O}\left(s^{-3}\right),
\end{align}
where $c(\alpha, \lambda)$ is an integration constant, independent of $s.$
\end{enumerate}
\end{theorem}

\begin{proof}
For small $s,$ we insert (\ref{Q22}) into (\ref{Q17}), then (\ref{Q20}) follows immediately. Similarly, for large $s,$ after substituting (\ref{Q19}) into (\ref{Q17}) followed by straightforward computations, (\ref{Q21}) is found. The hypothesis $\lambda+\alpha\notin \mathbb{Z}$ is appropriate in view of the form of the coefficients.
\end{proof}

Any monic polynomials orthogonal with respect to some weight has the Heine multiple-integral representation.
With the deformed Laguerre weight (\ref{AA1}), one finds,
\begin{align}\label{Q23}
P_{n}&(z; t, \alpha, \lambda)\nonumber\\
&=\frac{1}{D_{n}(t,\alpha,\lambda)}\frac{1}{n!}\int_{(0,\infty)^{n}}\prod_{q=1}^{n}(z-x_{q})\prod_{1\leq j<k\leq n}\left(x_{j}-x_{k}\right)^{2}\prod_{\ell=1}^{n}w(x_{\ell};t,\alpha,\lambda)dx_{\ell}.
\end{align}
Note $$P_n(z;t,\alpha,\lambda)=z^n+ \textsf{p}_1(n,t,\:\alpha,\:\lambda)\;z^{n-1}+\dots +P_n(0;t,\;\alpha,\:\lambda).$$
\\
It is immediate that the constant term of the monic orthogonal polynomial can be expressed by the quotient of Hankel determinants, that is,
\begin{equation}\label{Q24}
(-1)^{n}P_{n}\left(0;t,\alpha,\lambda\right)=\frac{D_{n}\left(t; \alpha+1,\lambda\right)}{D_{n}(t; \alpha,\lambda)}.
\end{equation}
If $t=0,$ then the constant term of the monic Laguerre polynomials, reads ((5.1.7), \cite{Szego1939}), with suitable modification,
\begin{equation}\label{QQ1}
(-1)^{n}P_{n}\left(0; 0, \alpha, \lambda\right)=\frac{\Gamma(n+1+\alpha+\lambda)}{\Gamma(1+\alpha+\lambda)}.
\end{equation}
From the above considerations, we state below a corollary which will ultimately help us to determine
the constant $c(\alpha, \lambda)$. 
\begin{cor}
Let $n \rightarrow \infty,$ $t\rightarrow 0,$ such that $s=4nt$ finite. Then
\begin{align}\label{Q25}
\lim_{n\rightarrow \infty}\frac{(-1)^{n}P_{n}\left(0;\frac{s}{4n}; \alpha,\lambda\right)}{(-1)^{n}P_{n}\left(0;0,\alpha,\lambda\right)}&=\frac{\mathbf{M}(s; \alpha+1,\lambda)}{\mathbf{M}(s; \alpha,\lambda)}\nonumber\\
&=\exp\Bigl(c_{1}(\alpha, \lambda)-\frac{\lambda}{2}\ln{s}-\frac{\lambda(1+2\alpha)}{2}s^{-\frac{1}{2}}-\frac{(1+2\alpha)\lambda^{2}}{4}s^{-1}\nonumber\\
&\quad +\frac{\lambda(1+2\alpha)(4\alpha(1+\alpha)-3-8\lambda^{2})}{48}s^{-\frac{3}{2}}+\mathcal{O}\left(s^{-2}\right)\Bigr),\end{align}
where $c_{1}(\alpha, \lambda)$ is another constant independent of $s$,
\begin{equation}\label{Q26}
c_{1}(\alpha, \lambda)=c(\alpha+1, \lambda)-c(\alpha, \lambda),
\end{equation}
and $c(\alpha, \lambda)$ is the constant in (\ref{Q21}).
\end{cor}

\begin{proof}
From (\ref{Q24}), one finds,
\begin{equation*}
\lim_{n\rightarrow \infty}\frac{(-1)^{n}P_{n}\left(0;\frac{s}{4n},\alpha,\lambda\right)}{(-1)^{n}P_{n}\left(0;0,\alpha,\lambda\right)}=
\lim_{n\rightarrow\infty}\frac{D_{n}(s/4n; \alpha+1,\lambda)}{D_{n}(s/4n; \alpha, \lambda)}\frac{D_{n}(0; \alpha,\lambda)}{D_{n}(0; \alpha+1, \lambda)}
=\frac{\mathbf{M}(s; \alpha+1,\lambda)}{\mathbf{M}(s; \alpha,\lambda)}.
\end{equation*}
Inserting $\mathbf{M}(s; \alpha,\lambda),$ (\ref{Q21}), into the above equation, one finds through straightforward computations
the equations (\ref{Q25}) and (\ref{Q26}) are found.
\end{proof}
With the aid of the formulas derived from the Dyson's Coulomb Fluid; see (\cite{ChenLawrence1998, Dyson1962},\cite{ChenIsmail1997}, \cite{CC2015, CCF2015}), the constant $c_{1}(\alpha, \lambda)$ can also be determined. In particular, for the deformed Laguerre weight, let ${\rm v}(x):=-\ln{w(x; t, \alpha, \lambda)}=x-\lambda\ln{(x+t)}-\alpha{\ln{x}},$ and consequently,
\begin{equation}\label{Q27}
{\rm v}'(x)=1-\frac{\lambda}{x+t}-\frac{\alpha}{x}.
\end{equation}
The equilibrium density is supported on $(a,b)$ where the endpoints satisfy the following supplementary conditions, which can be found, for instance, in
\cite{ChenIsmail1997, ChenLawrence1998},
\begin{equation}
\int_{a}^{b}\frac{x{\rm v}'(x)}{\sqrt{(b-x)(x-a)}}dx=2\pi{n},
\end{equation}
\begin{equation}
\int_{a}^{b}\frac{{\rm v}'(x)}{\sqrt{(b-x)(x-a)}}dx=0.
\end{equation}
 With the aid of the identities in \cite{YangHaqMcky2013}, \cite{DMY2014} and \cite{GR2007}, one finds,
\begin{equation}\label{Q28}
\frac{\lambda}{\sqrt{(t+a)(t+b)}}+\frac{\alpha}{\sqrt{ab}}-1=0,
\end{equation}
and
\begin{equation}\label{Q29}
\frac{\lambda{t}}{\sqrt{(t+a)(t+b)}}+\frac{a+b}{2}-\alpha-\lambda=2n.
\end{equation}
Let $X:=\sqrt{ab}$ denote the geometric mean of the end points, $a,\:b$.
Eliminating the center of mass $(a+b)/2$ from  (\ref{Q28}) and (\ref{Q29}), we deduce that $X$ satisfies the following quintic equation
\begin{equation}\label{Q30}
\left[t^{2}+t\left(2(2n+\alpha+\lambda-t)+\frac{2\alpha{t}}{X}\right)+X^{2}\right]\left(1-\frac{\alpha}{X}\right)^{2}=\lambda^{2}.
\end{equation}
From the quintic, we give an estimate of $X,$ and this will help us to derive the constant $c_{1}(\alpha, \lambda)$
in (\ref{Q26}). To proceed further, first note that for large $n$, we can approximate  $P_{n}(z)$ by
\begin{equation}\label{Q34}
P_{n}(z) \sim \exp\left[-S_{1}(z)-S_{2}(z)\right], \quad z\notin [a,b],
\end{equation}
where

$$S_{1}(z)={\frac{1}{4}}\ln \Bigl[ {\frac{16(z-a)(z-b)}{(b-a)^2}}\Bigl(\frac{ \sqrt{z-a}-\sqrt{z-b}}{
\sqrt{z-a}+\sqrt{z-b}}\Bigr)^2\Bigr]$$
and
$$S_2(z)=-N \ln \Bigl({\frac{ \sqrt{z-a}+\sqrt{z-b}}{2}}\Bigr)^2 +\int_a^b {\frac{dy}{2\pi}} {\frac{{\rm v}(y)}{\sqrt{(b-y)(y-a)}}}\Bigl[ {\frac{\sqrt{(z-a)(z-b)}}{y-z}}+1\Bigr]$$
are given by (4.7) and (4.8) in \cite{ChenLawrence1998}. Here we adopt
steps similar to those in \cite{CC2015}, and state the approximation in the following theorem.
\begin{theorem}
If ${\rm v}(x)=x-\lambda\ln{(x+t)}-\alpha{\ln{x}},$ $x>0,$ $t>0,$ $\alpha>0,$ then the value at $z=0$ of $P_{n}(z;t,\alpha,\lambda)$ is given by
\begin{align}\label{Q35}
&(-1)^{n}P_{n}\left(0;t,\alpha,\lambda\right)\sim (-1)^{n}\exp\left[-S_{1}(0;t,\alpha,\lambda)-S_{2}(0;t,\alpha,\lambda)\right]\nonumber\\
&\sim n^{n+\alpha+\lambda+\frac{1}{2}}e^{-n}\nonumber\\
&\times\exp\left[-\left(\alpha+\frac{1}{2}\right)\ln{\alpha}+\alpha-\frac{\lambda}{2}\ln{(4nt)}-\frac{\lambda(1+2\alpha)}{2}
\left(4nt\right)^{-\frac{1}{2}}-\frac{(1+2\alpha)\lambda^{2}}{4}\left(4nt\right)^{-1}\right];
\end{align}
this asymptotic estimate is valid for $n\rightarrow\infty,$ $t\rightarrow 0,$ such that $4nt$ is finite and large.
\end{theorem}
\begin{proof}
Inserting ${\rm v}(x)$ into ((4.2), \cite{CC2015}) and setting $z=0$, with the help of the integral identities in the \cite{YangHaqMcky2013}, \cite{DMY2014} and \cite{GR2007}
one finds,
\begin{align}\label{Q36}
\exp(-S_{2}&(0;t,\alpha,\lambda))\nonumber \\
\sim  &(-1)^{n}n^{n+\alpha+\lambda}e^{-n}\exp\left\{\sqrt{ab}-\alpha\ln{\sqrt{ab}}-\frac{\lambda}{2}\ln{(4nt)}-\lambda\ln\left[\frac{\sqrt{ab}}{2\sqrt{nt}}
+\left(1+\frac{ab}{4nt}\right)^{\frac{1}{2}}\right]\right\}.
\end{align}
With the expression $\exp(-S_{1}(z))$ ((4.3),\cite{CC2015}), by setting $z=0,$  one finds,
\begin{equation}\label{Q37}
\exp\left[-S_{1}(0;t,\alpha,\lambda)\right]\sim n^{\frac{1}{2}}\left(ab\right)^{-\frac{1}{4}}.
\end{equation}
To proceed further, we still need to get estimates on $a+b$ and $X=\sqrt{ab}.$ Sending $n\rightarrow \infty,$ $t\rightarrow 0,$
one finds (\ref{Q30}), is approximated by
\begin{equation}\label{Q51}
\left(X-\alpha\right)^{2}\left(4nt+X^{2}\right)-\lambda^{2}X^{2}=0.
\end{equation}
Hence,
\begin{equation}\label{Q32}
\sqrt{ab}=X\sim \alpha+\alpha{\lambda}(4nt)^{-\frac{1}{2}}+\alpha{\lambda^{2}}(4nt)^{-1}-
\frac{\alpha{\lambda}(\alpha^{2}-2\lambda^{2})}{2}(4nt)^{-\frac{3}{2}},
\end{equation}
and
\begin{equation}\label{Q33}
\frac{a+b}{2}\sim 2n+\alpha+\lambda.
\end{equation}
By these estimations, the asymptotic expansion for $P_{n}(0;t,\alpha,\lambda),$ namely $(\ref{Q35})$, follows immediately.
\end{proof}
\begin{rem3}
With the aid of the asymptotic formula,
\begin{equation}\label{AA9}
\frac{\Gamma(n+1+\alpha+\lambda)}{\sqrt{2\pi}}\sim n^{n+\alpha+\lambda+\frac{1}{2}}e^{-n},\quad {\rm for} \quad n\rightarrow\infty,
\end{equation}
we rewrite the asymptotic expression of $(-1)^{n}P_{n}(0;t,\alpha,\lambda)$ 
 as,
\begin{align}\label{Q38}
(-1)^{n}P_{n}(0;t,\alpha,\lambda)\sim &\frac{\Gamma(n+1+\alpha+\lambda)}{\Gamma(1+\alpha+\lambda)}\exp\left[\ln\left(\frac{\Gamma(1+\alpha+\lambda)}{\sqrt{2\pi}}\right)-
\left(\alpha+\frac{1}{2}\right)\ln{\alpha}+\alpha\right.\nonumber\\
&\left.-\frac{\lambda}{2}\ln{(4nt)}-\frac{\lambda(1+2\alpha)}{4}\left(4nt\right)^{-\frac{1}{2}}
-\frac{(1+2\alpha)\lambda^{2}}{4}\left(4nt\right)^{-1}\right].
\end{align}
Hence comparing with the exact evaluation of (\ref{Q25}), and the closed form of $(-1)^{n}P_{n}(0;0,\alpha,\lambda)$ given by (\ref{QQ1}), one finds,
\begin{align}\label{Q39}
\frac{P_{n}(0;t,\alpha,\lambda)}{P_{n}(0;0,\alpha,\lambda)}\sim \exp\left(c_{1}(\alpha, \lambda)-\frac{\lambda}{2}\ln(4nt)-\frac{(1+2\alpha)\lambda}{2}(4nt)^{-\frac{1}{2}}-
\frac{(1+2\alpha)\lambda^{2}}{4}\left(4nt\right)^{-1}\right),
\end{align}
from which $c_{1}(\alpha, \lambda)$ may be identified as
\begin{equation}
c_{1}(\alpha, \lambda)=\ln\left(\frac{\Gamma(1+\alpha+\lambda)}{\alpha}\right)-
\left[\left(\alpha-\frac{1}{2}\right)\ln\alpha-\alpha+\frac{1}{2}\ln\left(2\pi\right)\right].
\end{equation}
Noting the difference equation (\ref{Q26}), $c(\alpha, \lambda)$ is found to be
\begin{equation}\label{Q40}
c(\alpha, \lambda)=\ln\left(\frac{G(1+\alpha+\lambda)}{G(1+\alpha)}\right)-b(\alpha),
\end{equation}
where $G(z)$ is the Barnes G-function and $b(\alpha)$ satisfies the following difference equation,
\begin{equation*}
b(1+\alpha)-b(\alpha)=\int_{0}^{\infty}\left(\frac{1}{2}-\frac{1}{t}+\frac{1}{e^{t}-1}\right)\frac{e^{-\alpha{t}}}{t}dt.
\end{equation*}
The above is obtained with Binet's  formula, (Page 249, \cite{WW1958}), reproduced here,
\begin{equation*}
\ln{\Gamma(z)}=\left(z-\frac{1}{2}\right)\ln{z}-z+\frac{1}{2}\ln{(2\pi)}+\int_{0}^{\infty}\left(\frac{1}{2}-
\frac{1}{t}+\frac{1}{e^{t}-1}\right)\frac{e^{-zt}}{t}dt, \quad {\rm Re}z>0.
\end{equation*}
Note that the constant of (\ref{Q40}) was found from the double scaling analysis of the singularly  deformed Jacobi ensemble, \cite{CCF2015}. For large $s,$
\begin{align}\label{Q56}
\mathbf{\widetilde{M}}(s; \alpha, \lambda):&=\mathbf{M}(s; \alpha, \lambda)\exp\left(b(\alpha)\right)\nonumber\\
&=\exp\left(\ln\left(\frac{G(1+\alpha+\lambda)}{G(1+\alpha)}\right)+\lambda{s^{\frac{1}{2}}}-\frac{\lambda(2\alpha+\lambda)}{4}{\ln}s
-\frac{\lambda(4\alpha^{2}-1)}{8}s^{-\frac{1}{2}}\right.\nonumber\\
&\quad\left.-\frac{\lambda^{2}(4\alpha^{2}-1)}{16}s^{-1}+\frac{\lambda(4\alpha^{2}-1)(4\alpha^{2}-9-16\lambda^{2})}{384}s^{-\frac{3}{2}}
+\mathcal{O}\left(s^{-2}\right)\right).
\end{align}
\end{rem3}


\section{Asymptotic expansion from Normand's formula}
Normand's formula (\cite{N2004}), was applied in \cite{CCF2015, PFW2006} to obtain an asymptotic expansion of the Hankel determinant
around $t=0.$ Double scaling the variables, $n\rightarrow \infty,$ $t\rightarrow 0$ such that  $s=4nt$ finite, and combining with
Norman's formula, we find the small $s$ expansion for the current problem.
This reduces to ((1.22),\cite{TW1994}), with special choice of parameters.
\par
The Hankel determinant in terms of moment reads
\begin{align}\label{aa22}
D_{n}(t; \alpha,\lambda)
={\rm det}\left[\mu_{j+k}(t)\right]_{j,k=0}^{n-1},
\end{align}
where the moments $\mu_{k}(t)$ are
\bea\label{aa23}
\mu_{k}(t)=\int_{0}^{\infty}x^{k}(t+x)^{\lambda}x^{\alpha}e^{-x}dx, \quad k=0,1,\dots .
\eea
\par
The moments are expressed in terms of Kummer's function $U(a,b,z),$
from which we deduce the leading terms of the expansion of $\mu_k(t)$ for small $t$.
\begin{lemma}
The moment functions are given by $(\ref{aa23})$ and read
\begin{align}\label{aa24}
\mu_{k}(t)&=\int_{0}^{\infty}x^{k}(t+x)^{\lambda}x^{\alpha}e^{-x}dx\nonumber\\
&=\Gamma(k+\alpha+1)U(k+\alpha+1, k+\alpha+\lambda+2, t)t^{k+\alpha+\lambda+1},
\end{align}
and the series expansion of $\mu_{m}(t)$ around $t=0$ reads
\begin{align}\label{aa25}
\mu_{k}(t)=&\varphi_{k}(0)+t\varphi_{k}'(0)+t^{2}\frac{\varphi_{k}''(0)}{2!}+\mathcal{O}(t^{3})+t^{k+\alpha+\lambda+1}
\phi_{k}(0)(1+\mathcal{O}(t^{2}))\nonumber\\ &+\mathcal{O}\left(t^{2(k+\alpha+\lambda+1)}\right),
\end{align}
where
\bea\label{a98}
\varphi_{k}(0)=\Gamma(k+\alpha+\lambda+1),\quad \varphi_{k}'(0)=\lambda\Gamma(k+\alpha+\lambda),
\eea
\bea\label{a99}
\varphi_{k}''(0)=\lambda(\lambda-1)\Gamma(k+\alpha+\lambda-1),
\eea
and
\bea\label{a100}
\phi_{k}(0)=\frac{\Gamma(k+\alpha+1)\Gamma(-k-\alpha-\lambda-1)}{\Gamma(-\lambda)},
\eea
subject to ${\rm Re}(\alpha)>-1,$ $\lambda+\alpha \notin \mathbb{Z},$ $k\in \mathbb{N}$ and $|\arg{t}|<\frac{\pi}{2}.$
\end{lemma}
\begin{proof}
From the integral representation,
\begin{align}\label{aa97}
U(a,\;b,\;z)&=\frac{1}{\Gamma(a)}\int_{0}^{\infty}e^{-zx}x^{a-1}(1+x)^{b-a-1}dx,\nonumber\\
&=\frac{\Gamma(1-b)}{\Gamma(a-b+1)}\ _{1}\mathrm{F}_{1}(a,b,z)+\frac{\Gamma(b-1)}{\Gamma(a)}z^{1-b}\ _{1}\mathrm{F}_{1}(a-b+1,2-b,z),
\end{align}
where ${\rm Re}(a)>0,$ $|\arg{z}|<\frac{\pi}{2},$ see (equation 13.4.4,\cite{OLBC2010}), the second equality see (P1024 and P1025, \cite{GR2007}) and the series expansion around $z=0.$
\begin{equation*}
\ _{1}\mathrm{F}_{1}(a,b,z)=\sum_{k=0}^{\infty}\frac{(a)^{(k)}}{(b)^{(k)}}\frac{z^{k}}{k!},\end{equation*}
where
$$
\quad (a)^{(0)}=1, \quad (a)^{(n)}=a(a+1)\dots(a+n-1).
$$
The results follow from a minor change of parameters.
\end{proof}

\begin{theorem}
The Hankel determinant $D_{n}(t; \alpha, \lambda)$ defined by $(\ref{aa22})$ has an expansion around $t=0,$
\begin{align}\label{aa26}
D_{n}(t; \alpha, \lambda)&={\rm det}\left[\mu_{k+j}(t)\right]_{k,j=0}^{n-1}\nonumber\\
&=D_{n}(0; \alpha, \lambda)\left[1+\frac{\lambda{n}t}{\lambda+\alpha}+\frac{2\lambda{n}(n\lambda(\alpha+\lambda)-n-\alpha)t^{2}}
{4(\alpha+\lambda)((\alpha+\lambda)^{2}-1)}+\mathcal{O}(t^{3})\right.\nonumber\\ &\quad \left.-t^{\alpha+\lambda+1}\frac{\Gamma(\alpha+1)\Gamma(\lambda+1)
\Gamma(\alpha+\lambda+n+1)\sin{\lambda{\pi}}}{\Gamma(\alpha+\lambda+1)\Gamma^{2}
(\alpha+\lambda+2)\Gamma(n)\sin{(\alpha+\lambda)\pi}}(1+\mathcal{O}(t))+\mathcal{O}(t^{2(\alpha+\lambda+1)})\right],
\end{align}
with $D_{n}(0;\alpha,\lambda)$ given by (\ref{Q15A}). This is valid for ${\rm Re}(\alpha)>-1,$ $\lambda+\alpha \notin \mathbb{Z},$ and $|\arg{t}|<\frac{\pi}{2}.$
\end{theorem}
\begin{proof} In the following
$<t^{m}>f(t)$ denotes the coefficient of $t^{m}$ in the series expansion of $f(t)$ in $t$. By $(\ref{aa25}),$
\begin{align}\label{aa100}
D_{n}(t; \alpha, \lambda)&={\rm det}\left[\mu_{k+j}(t)\right]_{k,j=0}^{n-1} \nonumber\\
&\quad\sim {\rm det}\left[\varphi_{k+j}(0)+t\varphi_{k+j}'(0)+t^{2}
\frac{\varphi_{k+j}''(0)}{2!}+t^{k+j+\alpha+\lambda+1}\phi_{k+j}(0)\right]_{k,j=0}^{n-1}\nonumber\\
&\quad\sim {\rm det}\left[\varphi_{k+j}(0)\right]_{k,j=0}^{N-1}+t<t>{\rm det}\left[\varphi_{k+j}(0)+t\varphi_{k+j}'(0)\right]_{k,j=0}^{n-1}\nonumber\\
&\quad\quad+t^{2}<t^{2}>{\rm det}\left[\varphi_{k+j}(0)+t\varphi_{k+j}'(0)+t^{2}\frac{\varphi_{k+j}''(0)}{2!}\right]_{k,j=0}^{n-1}\nonumber\\
&\quad\quad+t^{\alpha+\lambda+1}<t^{\alpha+\lambda+1}>{\rm det}\left[\varphi_{k+j}(0)+t^{k+j+\alpha+\lambda+1}\phi_{k+j}(0)\right]_{k,j=0}^{n-1}.
\end{align}
 Here $\varphi_{k}(0)$, $\varphi_{k}'(0)$, $\varphi_{k}''(0)$ and $\phi_{k}(0)$ are given by $(\ref{a98})$, $(\ref{a99})$ and $(\ref{a100})$, respectively. To proceed further, we note,
\bea\label{a101}
{\rm det}\left[\Gamma(z_{k}+j)\right]_{k,j=0}^{n-1}=\prod_{k=0}^{n-1}\Gamma(z_{k})\prod_{0\leq \ell<k \leq n-1}(z_{k}-z_{\ell}),
\eea
as found in \cite{N2004} and the determinant is a linear functional of the $j^{th}$ column. Hence,
\begin{align*}<t>{\rm det}\left[\varphi_{k+j}(0)+t\varphi_{k+j}'(0)\right]_{k,j=0}^{n-1}=D_{n}(0; \alpha, \lambda)\frac{\lambda{n}\Gamma(\alpha+\lambda)}{\Gamma(\lambda+\alpha+1)},
\end{align*}
and the coefficients of $t^{2}$ and $t^{\alpha+\lambda+1}$ are obtained, but we do not reproduce the steps involved.
\end{proof}

\begin{cor}
The small $t$ expansion of $H_{n}(t)$ reads,
\begin{align}\label{aa28}
H_{n}(t)=&t\frac{d}{dt}\ln{\frac{D_{n}(t; \alpha, \lambda)}{D_{n}(0; \alpha, \lambda)}}=\frac{\lambda{n}t}{\lambda+\alpha}-\frac{\lambda{\alpha}(n+\alpha+\lambda)n}{(\alpha+\lambda)^2
\left((\alpha+\lambda)^2-1\right)}t^{2}+\mathcal{O}(t^{3})\nonumber\\ &-t^{\alpha+\lambda+1}\frac{\Gamma(\alpha+1)\Gamma(\lambda+1)\Gamma(\alpha+\lambda+n+1)\sin{\lambda{\pi}}}
{\Gamma^{2}(\alpha+\lambda+1)\Gamma(\alpha+\lambda+2)\Gamma(n)\sin{(\alpha+\lambda)\pi}}(1+\mathcal{O}(t))+\mathcal{O}(t^{2(\alpha+\lambda+1)}).
\end{align}
\end{cor}
Under double scaling,
we obtained the asymptotic behavior of the logarithmic derivative of the Hankel determinant for small $s.$
\begin{cor}
The series of $\mathcal{H}(s)$ given by $(\ref{Q10})$ around $s=0$ reads,
\begin{align}\label{aa29}
\mathcal{H}(s)&=\frac{\lambda{s}}{4(\lambda+\alpha)}-\frac{\lambda{\alpha}}{16(\alpha+\lambda)^2\left((\alpha+\lambda)^2-1\right)}s^{2}
+\mathcal{O}(s^{3})\nonumber\\ &\quad -\frac{s^{\alpha+\lambda+1}}{4^{\alpha+\lambda+1}}
\frac{\Gamma(\alpha+1)\Gamma(\lambda+1)\sin{\lambda{\pi}}}{\Gamma^{2}(\alpha+\lambda+1)
\Gamma(\alpha+\lambda+2)\sin{(\alpha+\lambda)\pi}}(1+\mathcal{O}(s))+\mathcal{O}(s^{2(\alpha+\lambda+1)}).
\end{align}
The expansion of the {\it double scaled} MGF is given by
\begin{align}\label{aa30}
\ln\mathbf{M}(s; \alpha, \lambda)=&\frac{\lambda}{4(\lambda+\alpha)}s-\frac{\lambda{\alpha}}{32(\alpha+\lambda)^2\left((\alpha+\lambda)^2-1\right)}s^{2}+
\mathcal{O}(s^{3})\nonumber\\ &-\frac{s^{\alpha+\lambda+1}}{4^{\alpha+\lambda+1}}\frac{\Gamma(\alpha+1)\Gamma(\lambda+1)\sin{\lambda{\pi}}}{(\alpha+\lambda+1)\Gamma^{2}
(\alpha+\lambda+1)\Gamma(\alpha+\lambda+2)\sin{(\alpha+\lambda)\pi}}(1+\mathcal{O}(s))\nonumber\\
&+\mathcal{O}(s^{2(\alpha+\lambda+1)}),
\end{align}
valid for ${\rm Re}(\alpha)>-1,$ $\lambda+\alpha \notin \mathbb{Z},$ and $|\arg{t}|<\frac{\pi}{2}.$
\end{cor}
\begin{proof}
Substituting $t=\frac{s}{4n}$ into $(\ref{aa26})$ and combining with $(\ref{Q3}),$ $(\ref{Q10}),$ followed by letting $n \rightarrow \infty$, we find the equation $(\ref{aa29})$.\par
With the aid of $(\ref{Q12}),$ $(\ref{aa26})$ and $t=\frac{s}{4n},$ letting $n \rightarrow \infty$, one easily arrives at $(\ref{aa30}).$
\end{proof}
\section{Rational solutions of the Painlev\'e equations.}

There are special polynomials associated with the rational solutions of the third Painlev\'e transcendent. These are given in the following theorem of Kajiwara and Masuda \cite{KM1999}, also considered by Clarkson \cite{Clarkson2003} and generalized by Clarkson \cite{PAC2005}. For convenience, we restate Theorem 3.4 in \cite{PAC2005} here.

\begin{theorem}\label{TH1}
Suppose that $S_{n}(\zeta;\mu)$ satisfies the recursion relation
\begin{equation}\label{E1}
S_{n+1}S_{n-1}=-\zeta\left[S_{n}\frac{d^{2}S_{n}}{d\zeta^{2}}-\left(\frac{dS_{n}}{d\zeta}\right)^2\right]
-S_{n}\frac{dS_{n}}{d\zeta}+(\zeta+\mu)S_{n}^{2},
\end{equation}
with $S_{-1}(\zeta; \mu)=S_{0}(\zeta; \mu)=1.$ Then
\bea\label{E2}
v(\zeta;a,b,1,-1):=1+\frac{d}{d\zeta}\left\{\ln\left[\frac{S_{n-1}(\zeta;\mu)}{S_{n}(\zeta;\mu+1)}\right]\right\}
=\frac{S_{n}(\zeta;\mu)S_{n-1}(\zeta;\mu+1)}{S_{n}(\zeta;\mu+1)S_{n-1}(\zeta;\mu)},
\eea
satisfies the third Painlev\'e transcendent ${P_{III}}(a_{n}, b_{n}, 1, -1),$
\bea\label{E2}
\frac{d^{2}v}{d\zeta^2}=\frac{1}{v}\left(\frac{dv}{d\zeta}\right)^{2}-\frac{1}{\zeta}\frac{dv}{d\zeta}+\frac{a_{n}v^{2}+b_{n}}{\zeta}+v^{3}-\frac{1}{v},
\eea
with $a_{n}:=2n+2\mu+1$ and $b_{n}:=2n-2\mu-1.$
\end{theorem}

Our strategy for solving the fifth Painlev\'e equation (\ref{Q13}) with special choices of the parameters
is to turn it into a particular third Painlev\'e equation or other equations solved previously. We state below a theorem which describe the connection between the $P_{III}$ and the $P_V$; this is Theorem 34.1 in \cite{GLS2002}.

\begin{theorem}\label{TH4}
Let $y=y(x)$ be a solution of the third Painlev\'e equation $P_{III}(a,b,1,-1)$
\begin{equation}\label{A01}
y''(x)=\frac{(y'(x))^2}{y(x)}-\frac{y'(x)}{x}+\frac{ay^{2}(x)+b}{x}+y^{3}(x)-\frac{1}{y(x)},
\end{equation}
such that
\begin{equation}\label{A02}
{\rm R}(x):=y'-{\varepsilon}{y^{2}}+\frac{(1-{\varepsilon}{a})y}{x}+1\neq 0,\quad {\rm with} \quad \varepsilon^{2}=1.
\end{equation}
Then the function
\begin{equation}\label{A03}
u(\tau):=1-\frac{2}{{\rm R}(\sqrt{2\tau})}
\end{equation}
is a solution of the fifth Painlev\'e transcendent,
\begin{equation}\label{A04}
u''(\tau)=\frac{3u-1}{2u(u-1)}\left(u'\right)^{2}-\frac{u'}{\tau}
+\frac{(u-1)^{2}}{\tau^{2}}\left(\alpha_{1}u+\frac{\beta_{1}}{u}\right)+\frac{\gamma_{1}}{\tau}u,
\end{equation}
with parameter values
\bea\label{AA13}
\left(\alpha_{1}, \beta_{1}, \gamma_{1}, 0\right)=\left(\frac{(b-\varepsilon{a}+2)^{2}}{32}, -\frac{(b+\varepsilon{a}-2)^{2}}{32}, -\varepsilon, 0\right).
\eea
\end{theorem}

With the help of the two theorems above, we study the \textit{rational solutions} of $P_{V}\left(\frac{\alpha^{2}}{2},-\frac{\lambda^{2}}{2},\frac{1}{2},0\right)$ (\ref{Q13}), and the corresponding $\sigma$-form given by (\ref{Q14}) by setting $\alpha=k+\frac{1}{2}, k\in \mathbb{N}$. The rational solutions of these Painlev\'e equations
are cast in terms of the special polynomials which satisfy (\ref{E1}).

\begin{theorem}\label{TH2}
Let $\alpha=k+\frac{1}{2},$ $k\in \mathbb{N},$ then rational solutions 
may be obtained as follows:
\begin{equation}\label{C11}
\mathcal{H}(s)=\frac{(sg'(s))^{2}}{4g(s)(g(s)-1)^{2}}-\frac{1}{4g(s)}\left[\lambda+\left(k+\frac{1}{2}\right)g(s)\right]^{2}
+\frac{sg(s)}{4(g(s)-1)},
\end{equation}
and
\begin{equation}\label{C12}
g(s)=1-\frac{2}{R(\sqrt{s})},
\end{equation}
with $R(x)$ given by
\bea\label{AA16}
R(x):=y'(x)+y^{2}(x)+\frac{\left(2k-2\lambda\right)y(x)}{x}+1,
\eea
and $y(x)$ in terms of special polynomials,
\bea\label{AA17}
y(x)=\frac{S_{k}(-x;\lambda+1)S_{k-1}(-x;\lambda)}{S_{k}(-x;\lambda)S_{k-1}(-x;\lambda+1)}.
\eea
Here $y(x)$ satisfies $P_{III}\left(2k-2\lambda-1, 2k+2\lambda+1, 1,-1\right),$ and $S_{k}(x)$ satisfies (\ref{E1}).
\end{theorem}

\begin{proof}
With the change of variable $\tau=\frac{s}{2}$ in the $P_V$ (\ref{A04}), then $u(s)$ satisfies another $P_V$ ,
\begin{equation}\label{A05}
u''(s)=\frac{3u(s)-1}{2u(s)(u(s)-1)}\left(u'(s)\right)^{2}-\frac{u'(s)}{s}
+\frac{(u(s)-1)^{2}}{s^{2}}\left(\alpha_{1}u+\frac{\beta_{1}}{u}\right)+\frac{\gamma_{1}}{2s}u
\end{equation}
with parameters given by (\ref{AA13}). Comparing the equations (\ref{A05}) and (\ref{Q13}), one obtains $\varepsilon=-1.$ The solutions of ${\rm P_{V}\left(\frac{\alpha^{2}}{2},-\frac{\lambda^{2}}{2},\frac{1}{2},0\right)}$ from (\ref{Q13}),  in terms of
 ${\rm P_{III}(a, b, 1, -1)}$ of (\ref{A01}), then (\ref{C12}) are obtained, and ${\rm R}$ is given by
\begin{equation}\label{AA15}
{\rm R}(x):=y'(x)+y^{2}(x)+\frac{(1+a)y(x)}{x}+1.
\end{equation}
Comparing equation (\ref{A05}) with (\ref{Q13}), and with our concrete problem, we choose $a$ and $b$ as follows,
\bea\label{AA14}
\left\{
\begin{array}{l}
a=2\alpha-2\lambda-2 \\[3mm]
b=2\alpha+2\lambda
\end{array}
\right.
\eea
although there are four situations.

Let $\alpha=k+\frac{1}{2},$ $k\in \mathbb{N},$ one gets $a=2k-2\lambda-1$ and $b=2k+2\lambda+1$ in (\ref{A01}), with the aid of Theorem \ref{TH1}, then we obtained (\ref{AA17}), (\ref{C11}). The equation (\ref{AA16}) is obtained from (\ref{Q15}), (\ref{AA15}).
\end{proof}

\begin{rem11}
Setting $g(s)=u^{-1}(\zeta),$ and $s=\zeta^{2},$ in $P_{V}\left(\frac{\alpha^{2}}{2},-\frac{\lambda^{2}}{2},\frac{1}{2},0\right),$ (\ref{Q13}), then we arrive at
the second order differential equation ((3.5), \cite{PAC2005}) with parameters given by ((3.7), \cite{PAC2005}). With the aid of the Theorem 3.5 in \cite{PAC2005}, we derive the rational solution $P_{V}\left(\frac{\alpha^{2}}{2},-\frac{\lambda^{2}}{2},\frac{1}{2},0\right),$ with $\alpha=n-\frac{1}{2}$ and $\lambda$ is arbitrary, as follows:
\begin{equation*}
g(s)=\frac{{-\lambda}S_{n-1}(\sqrt{s};-\lambda+1)S_{n-1}(\sqrt{s};-\lambda-1)}{S_{n}(\sqrt{s};-\lambda)S_{n-2}(\sqrt{s};-\lambda)}.
\end{equation*}
\end{rem11}

\subsection{${\rm P_{V}},$ the sine-Gordon equation and discrete Painlev\'e ${\rm II}.$}
The third Painlev\'e equation can be transformed into a sine-Gordon equation, as in Casini, Fosco
 and Huerta \cite{CFH2005}. A special case of $P_{III}$ was reduced to a discrete Painlev\'e {\rm II}
 by Periwal and Shevitz \cite{PS1990}, which can also be
found in \cite{TW1999}. With the help of these papers, we obtain solutions for special ${\rm P_{V}\left(\frac{\alpha^{2}}{2},-\frac{\lambda^{2}}{2},\frac{1}{2},0\right)},$ 
by setting $\alpha=0$ and $\lambda$ to be a positive integer, finally combining with a \textit{seed solution} of a particular $P_{III}$ and its four transformations. For completness, we present all four B\"acklund transformations. The first transformation leads to physically significant solutions as in Theorems \ref{TH5} and \ref{TH6}; we are not aware of physically significant implications of the other transformations.
The solutions of ${\rm P_{V}\left(\frac{\alpha^{2}}{2},-\frac{\lambda^{2}}{2},\frac{1}{2},0\right)},$ (\ref{Q13}), are obtained where  $\alpha$ and $\lambda$ are both integers.
\par
For convenience, we denote the solutions of ${\rm P_{V}\left(\frac{\alpha^{2}}{2},-\frac{\lambda^{2}}{2},\frac{1}{2},0\right)}$ in the style (\ref{Q13}) by
$g(s; \alpha, \lambda),$ and the
solutions of $\sigma$-form, (\ref{Q14}), by $\mathcal{H}(s; \alpha, \lambda).$

\begin{theorem} \label{TH5}
Let $\alpha=0,$ and $\lambda\in \mathbb{N}$, 
then the corresponding solutions are given by
\bea\label{B1}
\mathcal{H}(s; 0, \lambda)=\frac{(sg'(s; 0, \lambda))^{2}}{4g(s; 0, \lambda)(g(s; 0, \lambda)-1)^{2}}-\frac{\lambda^{2}}{4g(s; 0, \lambda)}+\frac{sg(s; 0, \lambda)}{4(g(s; 0, \lambda)-1)},
\eea
where $g(s; 0, \lambda)$ satisfies the difference equation in $\lambda,$
\bea\label{B2}
\frac{1}{\sqrt{1-g(s; 0, \lambda+1)}}+\frac{1}{\sqrt{1-g(s; 0, \lambda-1)}}=-\frac{2\lambda\sqrt{1-g(s; 0, \lambda)}}{\sqrt{s}g(s; 0, \lambda)},
\eea
with the initial condition $g(s; 0, 0)=0.$ For $\lambda=1,$ $g(s; 0, 1)$  is given by
\bea\label{B3}
g(s; 0, 1)=1-\frac{I_{0}^{2}(\sqrt{s})}{I_{1}^{2}(\sqrt{s})},
\eea
where $I_{k}(\cdot)$ is the modified Bessel function of the first kind of order $k.$
Iterating forward in $\lambda$ in (\ref{B2}) and (\ref{B3}), one obtains exact solutions of ${\rm P_{V}\left(\frac{\alpha^{2}}{2},-\frac{\lambda^{2}}{2},\frac{1}{2},0\right)},$for positive integral $\lambda.$
\end{theorem}

\begin{proof}
 Put $\alpha=0,$ 
 then $g(s; 0, \lambda)$ satisfies
\begin{equation}\label{E3}
\frac{d^2g}{ds^2}=\frac{3g-1}{2g(g-1)}\left(\frac{dg}{ds}\right)^2-\frac{1}{s}\frac{dg}{ds}-\frac{\lambda^2(g-1)^2}{2s^2g}+\frac{g}{2s}.
\end{equation}
Making a transformation and a change of variable;
\begin{equation}\label{B4}
\phi(r; \lambda)=\frac{1}{\sqrt{1-g(s; 0, \lambda)}}, \quad {\rm and} \quad r=\sqrt{s},
\end{equation}
we see that $\phi(r; \lambda)$ satisfies a $2$-dimensional sine-Gordon equation in the radial coordinates,
\begin{equation}\label{E4}
\frac{d^{2}\phi(r; \lambda)}{dr^{2}}+\frac{1}{r}\frac{d\phi(r; \lambda)}{dr}+\frac{\phi(r; \lambda)}{1-\phi^{2}(r; \lambda)}\left[\left(\frac{d\phi(r; \lambda)}{dr}\right)^{2}-\frac{\lambda^{2}}{r^{2}}\right]+\phi(r; \lambda)\left(1-\phi^{2}(r; \lambda)\right)=0.
\end{equation}
In addition, $\phi(r; \lambda)$ also satisfies differential-difference equations in $\lambda$,
\bea\label{E5}
\frac{d\phi(r; \lambda)}{dr}+\frac{\lambda}{r}\phi(r; \lambda)-(1-\phi^{2}(r; \lambda))\phi(r; \lambda-1)=0,
\eea
\bea\label{E6}
\frac{d\phi(r; \lambda-1)}{dr}-\frac{\lambda-1}{r}\phi(r; \lambda-1)+(1-\phi^{2}(r; \lambda-1))\phi(r; \lambda)=0.
\eea
Replacing $\lambda$ by $\lambda+1$ in (\ref{E6}), followed eliminating the first order derivative,
 we deduce that $\phi(r; \lambda)$ satisfies a `difference' equation with respect to $\lambda,$
\begin{equation}\label{E11}
\phi(r; \lambda+1)+\phi(r; \lambda-1)=\frac{2\lambda}{r}\frac{\phi(r; \lambda)}{1-\phi^{2}(r; \lambda)}.
\end{equation}
If $\lambda$ is an integer, 
then this equation is known as discrete Painlev\'e ${\rm II},$ as obtained by Periwal and Shavitz \cite{PS1990}, and Hisakado \cite{H1996}. Equations (\ref{E5}) and (\ref{E6}) are also obtained by Hisakado \cite{H1996}. These appeared Tracy and Widom \cite{TW1999} in the
study of connections between the characteristic function of the unitary matrices ensemble and the distribution functions of the length of the longest increasing subsequence. In Barashenkov and Pelinovsky \cite{BP1998}, they appeared in the construction of  explicit multi-vortex solutions of the complex sine-Gordon equation.
\par
From (\ref{B4}) and (\ref{E11}), it is easy to find that the (\ref{B2}), with boundary condition $g(s; \alpha, 0)=0,$ $(\ref{B4})$ and the differential-difference equation (\ref{E5}), has the solution,
\bea\label{E12}
\phi(r; 0)=1, \qquad \phi(r; 1)=\frac{I_{1}(r)}{I_{0}(r)}.
\eea
\end{proof}

\par
If $\lambda$ is an integer, then $\phi(r; \lambda)$ is related to a special case of $P_{III}$ obtained by Hisakado \cite{H1996}. See also Tracy and Widom \cite{TW1999}. We restate this connection as a Theorem  below.
\begin{theorem}\label{TH3}
If $\phi(r; \lambda)$ satisfies (\ref{E5}) and (\ref{E6}), then
\bea\label{E7}
W(r):=\frac{\phi(r; \lambda)}{\phi(r; \lambda-1)},
\eea
satisfies
\bea\label{E8}
\frac{d^{2}W}{dr^{2}}=\frac{1}{W}\left(\frac{dW}{dr}\right)^{2}-\frac{1}{r}\frac{dW}{dr}
+\frac{2(1-\lambda)}{r}W^{2}+\frac{2\lambda}{r}+W^{3}-\frac{1}{W},
\eea
which is a special case of the third Painlev\'e transcendent, ${\rm P_{III}}(-2\lambda+2, 2\lambda, 1, -1).$
\end{theorem}

 Translating $\lambda$ to $\lambda+1$ in Theorem \ref{TH3}, we obtain solutions of ${\rm P_{V}\left(\frac{\alpha^{2}}{2},-\frac{\lambda^{2}}{2},\frac{1}{2},0\right)},$ (\ref{Q13}), with $\alpha=1$ and $\lambda\in \mathbb{N}.$
We state this result as a theorem in the following.

\begin{theorem} \label{TH6}
Put $\alpha=1$ and $\lambda\in \mathbb{N},$ in  ${\rm P_{V}\left(\frac{\alpha^{2}}{2},-\frac{\lambda^{2}}{2},\frac{1}{2},0\right)},$ which has the $\sigma$-form (\ref{Q14}). Then the closed form solutions can be expressed  in terms of the modified Bessel function of the first kind. That is
\begin{equation}\label{C01}
\mathcal{H}(s; 1, \lambda)=\frac{(sg'(s; 1, \lambda))^{2}}{4g(s; 1, \lambda)(g(s; 1, \lambda)-1)^{2}}-\frac{\left(\lambda+g(s; 1, \lambda)\right)^{2}}{4g(s; 1, \lambda)}+\frac{sg(s; 1, \lambda)}{4(g(s; 1, \lambda)-1)},
\end{equation}
where
\begin{align}\label{C02}
g(s; 1, \lambda)&=1-\frac{2}{R(\sqrt{s})},\\
R(x)&:=y'(x)+y^{2}(x)+\frac{\left(1-2\lambda\right)y(x)}{x}+1,\\
\label{(5.100)}y(x)&=\frac{\phi(x; \lambda+1)}{\phi(x; \lambda)}.\end{align}
Here $\phi(x; \lambda)$ satisfies 
(\ref{E11}); the initial conditions read
\bea\label{C03}
\phi(x; 0)=1, \qquad \phi(x; 1)=\frac{I_{1}(x)}{I_{0}(x)}.
\eea
\end{theorem}

\begin{proof}
Translating $\lambda$ to $\lambda+1$ in (\ref{E7}), then
\bea\label{E9}
\widehat{W}(r)=\frac{\phi(r; \lambda+1)}{\phi(r; \lambda)},
\eea
is a solution of ${\rm P_{III}}(-2\lambda, 2\lambda+2, 1, -1)$, a special ${\rm P_{III}}$ (\ref{A01})
with $a,$ $b$ given by (\ref{AA14}) and $\alpha=1$, so (\ref{(5.100)}) is obtained. The initial data (\ref{C03}) are from (\ref{E12}). With the help of Theorem \ref{TH4} and Theorem \ref{TH5}, we obtain (\ref{C01}) and (\ref{C02}).
\end{proof}
To investigate the solution of $P_{V}(\frac{\alpha^2}{2}, -\frac{\lambda^2}{2},\frac{1}{2},0)$,
 where $\alpha \in \mathbb{N}$ and $\lambda \in \mathbb{N},$ we introduce four transformations.
\begin{lemma} \label{Le2}
\begin{enumerate}
$\bullet$ If $y(x)$ is a solution of the third Painlev\'e transcendent $P_{III}(a, b, 1, -1),$ (\ref{A01}), and the transformation $\mathfrak{K}_{1}$ is defined by
\begin{equation}\label{E15}
\mathfrak{K}_{1}:\left(x, y\right) \rightarrow \left(-x, h\right)=\left(-x, -\frac{xy'(x)-xy^{2}(x)-by(x)-y(x)-x}{y(x)\left[xy'(x)-xy^{2}(x)+ay(x)+y(x)-x\right]}\right),
\end{equation}
subject to the following condition,
\begin{equation}\label{E16}
y'(x)-y^{2}(x)+\frac{1+a}{x}y(x)-1 \neq 0,
\end{equation}
then $h(x)$ satisfies $P_{III}(2+a, 2+b, 1, -1).$
\par
$\bullet$ If the transformation $\mathfrak{K}_{2}$ is given by,
\begin{align}\label{A5}
\mathfrak{K}_{2}:\left(x, y\right) \rightarrow \left(-x, h\right)=\left(-x, -\frac{xy'(x)+xy^{2}(x)+(-1+b)y(x)+x}{y(x)\left[xy'(x)+xy^{2}(x)+(1-a)y(x)+x\right]}\right),
\end{align}
subject to the following condition,
\begin{equation*}
y'(x)+y^{2}(x)+\frac{1-a}{x}y(x)+1 \neq 0,
\end{equation*}
then $h(x)$ satisfies $P_{III}(-2+a, -2+b, 1, -1).$
\par
$\bullet$ If the transformation $\mathfrak{K}_{3}$ is defined by
\begin{align}\label{A6}
\mathfrak{K}_{3}:\left(x, y\right) \rightarrow \left(-x, h\right)=\left(-x, \frac{xy'(x)+xy^{2}(x)-by(x)-y(x)-x}{y(x)\left[xy'(x)+xy^{2}(x)-ay(x)+y(x)-x\right]}\right),
\end{align}
with restriction of,
\begin{equation*}
y'(x)+y^{2}(x)+\frac{1-a}{x}y(x)-1 \neq 0,
\end{equation*}
then $h(x)$ satisfies $P_{III}(a-2, b+2, 1, -1).$
\par
$\bullet$ If the transformation $\mathfrak{K}_{4}$ is given by
\begin{align}\label{A7}
\mathfrak{K}_{4}:\left(x, y\right) \rightarrow \left(-x, h\right)=\left(-x, \frac{xy'(x)-xy^{2}(x)+by(x)-y(x)+x}{y(x)\left[xy'(x)-xy^{2}(x)+ay(x)+y(x)+x\right]}\right),
\end{align}
subject to,
\begin{equation*}
y'(x)-y^{2}(x)+\frac{1+a}{x}y(x)+1 \neq 0,
\end{equation*}
then $h(x)$ satisfies $P_{III}(a+2, b-2, 1, -1).$
\end{enumerate}
\end{lemma}

\begin{proof}
By (\ref{E15}) 
one finds,
\begin{equation}\label{E18}
y'(x)=y^{2}(x)+\frac{(2+a+b)y(x)}{x+xh(x)y(x)}-\frac{(1+a)y(x)}{x}+1.
\end{equation}
Replacing $x$ by $-x$ in ${\rm P_{III}(a, b, 1, -1)},$ the equation (\ref{A01}) becomes,
\begin{equation*}
y''(x)=\frac{(y'(x))^{2}}{y(x)}-\frac{y'(x)}{x}-\frac{ay(x)^{2}+b}{x}+y^{3}(x)-\frac{1}{y(x)}.
\end{equation*}
Now substitute (\ref{E18}) into the above equation, we find
\begin{equation*}
h'(x)=\frac{y(x)h^{3}(x)}{1+h(x)y(x)}+\frac{x+(1+a)y(x)}{x+xh(x)y(x)}h^{2}(x)
+\frac{-(1+b)+xy(x)}{x+xh(x)y(x)}h(x)+\frac{1}{1+h(x)y(x)}.
\end{equation*}
With the above equation, we verify that $h(x)$ satisfies ${\rm P_{III}(2+a, 2+b, 1, -1)}.$ The proofs of other three are similar.
\end{proof}
We would like to point out that these transformations are different from those in \cite{MCB1995}; for more about the B\"{a}cklund transformations
of $P_{III},$ see \cite{Murata1995} and \cite{OLBC2010}.
\par
The seed solution obtained in Theorem \ref{TH6}, combined with the transformation $\mathfrak{K}_{1},$ gives us all the
solutions of $P_{V}\left(\frac{\alpha^{2}}{2},-\frac{\lambda^{2}}{2},\frac{1}{2},0\right),$ 
where $\alpha$ and $\lambda$ are positive integers, in terms of $a$ and $b$ given by (\ref{AA14}).

\begin{theorem}
Let $\alpha \in \mathbb{N},$ $\lambda \in \mathbb{N},$ in the fifth Painlev\'e equation
$P_{V}\left(\frac{\alpha^{2}}{2},-\frac{\lambda^{2}}{2},\frac{1}{2},0\right),$
then its solution reads
\begin{equation*}
g(s; \alpha, \lambda)=1-\frac{2}{R_{\alpha}(\sqrt{s})},
\end{equation*}
with $R_{\alpha}(x)$ given by
\begin{equation*}
R_{\alpha}(x):=\frac{d}{dx}\left(\mathfrak{K}^{\alpha-1}_{1}[y](x)\right)
+\left(\mathfrak{K}^{\alpha-1}_{1}[y](x)\right)^{2}+\frac{\left(2\alpha-2\lambda-1\right)\mathfrak{K}^{\alpha-1}_{1}[y](x)}{x}+1.
\end{equation*}
Here $\mathfrak{K}^{k}_{1}[f](x)$ denotes applying the transformation $\mathfrak{K}_{1}$  $k$ times to the function $f(x),$ where $y(x)$ is given by (\ref{(5.100)}). 
\end{theorem}

\begin{proof}
We take the solution (\ref{(5.100)}) of ${\rm P_{III}}(-2\lambda, 2\lambda+2, 1, -1),$ $\lambda \in \mathbb{N},$
as a seed solution. Apply  the transformation
$\mathfrak{K}_{1}$ $\alpha-1$ ($\alpha \in \mathbb{N}$) times to prove the Theorem.
\end{proof}

\section{Asymptotic behavior of the Double Scaled Moment Generating Function.}
In this section, we focus our attention on the (double-scaled) moment generating function for large and small $s.$
From these asymptotic expansions, we find that the distribution of the outage probability deviates significantly
from Gaussian. This has been partially observed in Chen and McKay \cite{YangMcky2012}. Forrester and Witte \cite{FW} (3.29), (6.10)  have otained asymptotic expansions of the logarithmic derivative of the isomonodromic $\tau$ function associated with the Laguerre ensemble; this $\tau$ function is given by a Hankel determinant analogous to our ( \ref{Q2}).  Similar series appear also in \cite{FO2010}, and see section 37 of  \cite{GLS2002} for details regarding $P_V$.  We compute and compare the first three or four cumulants.
\\
\par
For small $s,$ there is a Puiseux expansion of ${\mathcal{H}}(s)$ from (\ref{Q14}), which when combined with $(\ref{Q3}),$ $(\ref{Q1}),$ $(\ref{Q10}),$ and $\alpha \in \mathbb{Z}_{+},$ gives
\begin{align}\label{B5}
&\ln\mathbf{M}(s; \alpha, \lambda)\nonumber\\
&=\frac{\lambda}{4(\alpha+\lambda)}s-\frac{\alpha{\lambda}}{32(\alpha+\lambda)^{2}[(\alpha+\lambda)^{2}-1]}s^{2}
+\frac{\alpha{\lambda}}{96(\alpha+\lambda)^{3}[(\alpha+\lambda)^{2}-1][(\alpha+\lambda)^{2}-4]}s^{3}+\mathcal{O}\left(s^{4}\right)\nonumber\\
&\quad-\frac{C_{3}}{\alpha+\lambda+1}s^{\alpha+\lambda+1}\left(1+\frac{\alpha-\lambda}{16(\alpha+\lambda)^{2}(\alpha+\lambda-1)(\alpha+\lambda+2)}s+\mathcal{O}\left(s^{2}\right)\right)\nonumber\\
&\quad-\frac{C_{3}^{2}}{2(\alpha+\lambda+1)^{2}}s^{2(\alpha+\lambda+1)}\left(1+\frac{(\alpha-\lambda)(\alpha+\lambda+1)}{(\alpha+\lambda)(\alpha+\lambda+2)^{2}}s+\mathcal{O}\left(s^{2}\right)\right)\nonumber\\
&\quad-\frac{C_{3}^{3}}{3(\alpha+\lambda+1)^{3}}s^{3(\alpha+\lambda+1)}\left(1+\mathcal{O}\left(s\right)\right)+\mathcal{O}\left(s^{4(\alpha+\lambda+1)}\right),
\end{align}
with $C_{3}$ given by
\begin{equation*}
C_{3}=(-1)^{\alpha}\frac{\Gamma(1+\alpha)\Gamma(1+\lambda)}{4^{\alpha+\lambda+1}\Gamma^{2}(\alpha+\lambda+1)\Gamma(\alpha+\lambda+2)}.
\end{equation*}

$\bullet$ The cumulant generating function has an expansion $(\ref{a3}),$ so we list below four cumulants:
\begin{align*}
\kappa_{1}=&\frac{1}{4\alpha}s-\frac{1}{32\alpha(\alpha^{2}-1)}s^{2}
+\frac{1}{96\alpha^{2}(\alpha^{2}-1)(\alpha^{2}-4)}s^{3}+\mathcal{O}\left(s^{4}\right)\nonumber\\
&+(-1)^{\alpha}\frac{2+(\alpha+1)(\gamma+\ln4-\ln{s}+3\psi(0, \alpha+1))}{4^{\alpha+1}(\alpha+1)^{3}\Gamma^{2}(1+\alpha)}s^{1+\alpha}\left(1+\mathcal{O}\left(s\right)\right)\nonumber\\
&+\frac{2+(1+\alpha)(\gamma+\ln4-\ln{s}+3\psi(0, \alpha+1))}{16^{1+\alpha}(1+\alpha)\Gamma^{4}(\alpha+2)}s^{2(1+\alpha)}\left(1+\mathcal{O}\left(s\right)\right)
+\mathcal{O}\left(s^{3(1+\alpha)}\right),\\
\kappa_{2}=&-\frac{s}{2\alpha^{2}}+\frac{2\alpha^{2}-1}{8\alpha^{2}(\alpha^{2}-1)^{2}}s^{2}
-\frac{7\alpha^{4}-25\alpha^{2}+12}{48\alpha^{3}(\alpha^{2}-1)^{2}(\alpha^{2}-4)^{2}}s^{3}+\mathcal{O}\left(s^{4}\right)\nonumber\\
&+\left(C_{4}+6(1+\alpha)[4+(1+\alpha)(2\gamma+\ln{16})]\ln{s}-36(1+\alpha)^{2}\psi(0, \alpha+1)\ln{s}-6(1+\alpha)^{2}(\ln{s})^{2}\right.\nonumber\\
&\quad\left.+36(1+\alpha)(2+(1+\alpha)(\gamma+\ln{4}))\psi(0, \alpha+1)+54(1+\alpha)^{2}\psi(0, \alpha+1)^{2}\right.\nonumber\\
&\quad\left.-18(1+\alpha)^{2}\psi(1, \alpha+1)
\right)\frac{2(-1)^{1+\alpha}}{3(\alpha+1)^{4}4^{\alpha+2}\Gamma^{2}(1+\alpha)}s^{1+\alpha}\left(1+\mathcal{O}\left(s\right)\right)+\mathcal{O}\left(s^{2(1+\alpha)}\right),\\
\kappa_{3}=&\frac{3}{2\alpha^{3}}s-\frac{9-27\alpha^{2}+30\alpha^{4}}{16\alpha^{3}(\alpha^{2}-1)^{3}}s^{2}+\frac{96-340\alpha^{2}+447\alpha^{4}-195\alpha^{6}+28\alpha^{8}}{16\alpha^{4}(\alpha^{2}-1)^{3}(\alpha^{2}-4)^{3}}s^{3}+\mathcal{O}\left(s^4\right)+\mathcal{O}\left(s^{1+\alpha}\right),\\
\kappa_{4}=&-\frac{6}{\alpha^{4}}s+\frac{-3+12\alpha^{2}-18\alpha^{4}+15\alpha^{6}}{\alpha^{4}(\alpha^{2}-1)^{4}}s^{2}\nonumber\\
&-\frac{640-3120\alpha^{2}+6240\alpha^{4}-6655\alpha^{6}+3450\alpha^{8}-855\alpha^{10}+84\alpha^{12}}{4\alpha^{5}(\alpha^{2}-1)^{4}(\alpha^{2}-4)^{4}}s^{3}+\mathcal{O}\left(s^4\right)+\mathcal{O}\left(s^{1+\alpha}\right);
\end{align*}
here the $(-1)^{\alpha}$ arises from $\sin \lambda\pi /\sin (\lambda +\alpha )\pi$ in (\ref{aa29}). Also$C_{4}$ is given by
\begin{align*}
C_{4}&=36+(1+\alpha)^{2}(6\gamma^{2}+\pi^{2}+24(\ln2)^{2})+24(1+\alpha)(\gamma+2\ln{2}+\gamma(1+\alpha))\ln2\nonumber\\
&\approx 116.126+113.128\alpha+33.0018\alpha^{2},
\end{align*}
and
\begin{equation}\label{A16}
\psi(n,x)=\frac{d^{n}\psi(x)}{dx^{n}},  \quad \psi(x)=\frac{\Gamma'(x)}{\Gamma(x)},
\end{equation}
$\gamma=0.57721\dots$ is Euler's constant and $\zeta(\cdot)$ is Riemann's Zeta function.

\par
$\bullet$ For small $s,$ and $\alpha=0$ in (\ref{B5}), we find,
\begin{align}\label{B6}
\ln\mathbf{M}(s; 0, \lambda)
&=\frac{1}{4}s-\frac{1}{4^{\lambda+1}\Gamma^{2}(\lambda+2)}s^{\lambda+1}
\left(1-\frac{1}{16\lambda(\lambda-1)(\lambda+2)}s+\mathcal{O}\left(s^{2}\right)\right)\nonumber\\
&\quad-\frac{1}{2^{4\lambda+5}\Gamma^{4}(\lambda+2)}s^{2(\lambda+1)}\left(1-\frac{\lambda+1}
{(\lambda+2)^{2}}s+\mathcal{O}\left(s^{2}\right)\right)+\mathcal{O}\left(s^{3(\lambda+1)}\right).
\end{align}
We list here four cumulants:

\begin{align*}
\kappa_{1}&=\frac{s}{4}\left(2-2\gamma+\ln4\right)-\frac{s\ln{s}}{4}+\mathcal{O}
\left(s^{2}\right)+\mathcal{O}\left(s^{2}\ln{s}\right)+\mathcal{O}\left(s\ln^{2}{s}\right),\\
\kappa_{2}&=c_{21}s+c_{22}s{\ln{s}}-\frac{1}{4}s(\ln{s})^{2}+\mathcal{O}\left(s^{2}\right)+\mathcal{O}\left(s^{2}\ln{s}\right)+\mathcal{O}\left(s\ln^{3}{s}\right),\\
\kappa_{3}&=c_{31}s+c_{32}s\ln{s}+c_{33}s(\ln{s})^{2}-\frac{1}{4}s(\ln{s})^{3}+\mathcal{O}\left(s^{2}\right)+\mathcal{O}\left(s^{2}\ln{s}\right)+\mathcal{O}\left(s\ln^{3}{s}\right),\\
\kappa_{4}&=c_{40}s+c_{41}s\ln{s}+c_{42}s(\ln{s})^2+c_{43}s(\ln{s})^3-\frac{1}{4}s(\ln{s})^4++\mathcal{O}\left(s^{2}\right)+\mathcal{O}\left(s^{2}\ln{s}\right)+\mathcal{O}\left(s\ln^{5}{s}\right),
\end{align*}
\noindent where $c_{21}$ and $c_{22}$ in $\kappa_2$ are given by
\begin{align*}
c_{21}&=\frac{1}{12}\left[-12\gamma^{2}+\pi^{2}+12\gamma\left(2+\ln{4}\right)-3\left(6+\ln{4}\left(4+\ln{4}\right)\right)\right]\approx -0.923,\\
c_{22}&=-\gamma+1+\ln{2} \approx 1.116;
\end{align*}

\noindent where $c_{31},$ $c_{32},$ $c_{33}$ and $c_{34}$ in $\kappa_3$ are given by
\begin{align*}
c_{31}&=\frac{1}{4}\left[24-8\gamma^{3}+12\gamma^{2}(2+\ln{4})-\pi^{2}(2+\ln{4})+\ln{4}\left(18+\ln{4}(6+\ln{4})\right)\right.\nonumber\\
&\quad\left.+2\gamma(\pi^{2}-3(6+\ln{4}(4+\ln{4})))-4\zeta(3)\right]\approx 0.418,\nonumber\\
c_{32}&=\frac{1}{4}\left[-12\gamma^{2}+\pi^{2}+12\gamma(2+\ln{4})-3(6+\ln{4}(4+\ln{4}))\right] \approx -0.692,\nonumber\\
c_{33}&=\frac{1}{2}\left[3-3\gamma+\ln{8}\right] \approx 0.419,
\end{align*}

\noindent and where $c_{40},$ $c_{41},$ $c_{42},$ $c_{43}$ and $c_{44}$ in $\kappa_4$ are given by
\begin{align*}
c_{40}=&-30-4\gamma^{4}-\frac{\pi^{4}}{20}+8(2+\ln{4})\gamma^{3}+\pi^{2}\left(3+2\ln{4}+\frac{(\ln{4})^{2}}{2}\right)+2\gamma^{2}\left(\pi^{2}-18-12\ln{4}-3(\ln{4})^{2}\right)\nonumber\\
&-\ln{4}\left(24+9\ln{4}+2(\ln{4})^2+\frac{1}{4}(\ln{4})^{3}-4\zeta(3)\right)+2\gamma\left[24-2\pi^{2}(1+\ln{2})+36\ln{2}\right.\nonumber\\
&\left.+6(\ln{4})^{2}+(\ln{4})^{3}-4\zeta(3)\right]+8\zeta(3)\approx 4.238,\nonumber\\
c_{41}=&24-36\gamma+24\gamma^{2}-8\gamma^{3}-2\pi^{2}+2\gamma\pi^{2}+18\ln{4}-24\gamma\ln{4}+12\gamma^{2}\ln{4}-\pi^{2}\ln{4}\nonumber\\
&+6(\ln{4})^{2}-6\gamma(\ln{4})^{2}+(\ln{4})^{3}-4\zeta(3)\approx 1.673,\nonumber\\
c_{42}=&-9+12\gamma-6\gamma^{2}+\frac{\pi^{2}}{2}-6\ln{4}+6\gamma\ln{4}-\frac{3}{2}(\ln{4})^{2}\approx -5.537,\nonumber\\
c_{43}=&2-2\gamma+\ln{4} \approx 2.232.
\end{align*}

\par
For small $s,$ and $\alpha=1$ in (\ref{B5}), we find,
\begin{align}\label{BB6}
&\ln\mathbf{M}(s; 1, \lambda)\nonumber\\
=&\frac{\lambda}{4(1+\lambda)}s-\frac{1}{32(2+\lambda)(1+\lambda)^{2}}s^{2}+\frac{1}{96(\lambda-1)(\lambda+1)^{3}(2+\lambda)(3+\lambda)}s^{3}+\mathcal{O}\left(s^{4}\right)\nonumber\\
&+\frac{1}{(\lambda+1)4^{\lambda+2}\Gamma^{2}(\lambda+3)}s^{\lambda+2}\left(1+\frac{1-\lambda}{16\lambda(\lambda+1)(\lambda+3)}s+\mathcal{O}\left(s^{2}\right)\right)\nonumber\\
&-\frac{1}{2^{9+4\lambda}(\lambda+1)^{2}(\lambda+2)^{2}\Gamma^{2}(\lambda+1)\Gamma^{2}(\lambda+3)}s^{2\lambda+4}\left(1+\frac{(1-\lambda)(\lambda+2)}{(\lambda+1)^{3}(\lambda+3)^{2}}s+\mathcal{O}\left(s^{2}\right)\right)\nonumber\\
&+\mathcal{O}\left(s^{3(\lambda+2)}\right).
\end{align}
We list here three cumulants:
\begin{align*}
\kappa_{1}&=\frac{1}{4}s+\left(-\frac{3}{128}+\frac{1}{32}\gamma-\frac{1}{32}\ln{2}\right)s^2+\frac{1}{64}s^{2}\ln{s}
+\mathcal{O}\left(s^{3}\right)+\mathcal{O}\left(s^{2}\ln^{2}{s}\right)+\mathcal{O}\left(s^{3}\ln{s}\right),\\
\kappa_{2}&=-\frac{1}{2}s+\left[\frac{11}{64}-\frac{\gamma}{4}+\frac{\gamma^{2}}{16}-\frac{\pi^{2}}{192}+\frac{\ln{2}}{4}+\frac{(\ln{2})^{2}}{16}-\frac{\gamma{\ln{2}}}{8}\right]s^{2}+\frac{1}{16}\left(\gamma-2-\ln{2}\right)s^{2}\ln{s}\nonumber\\
&\quad+\frac{1}{64}s^{2}(\ln{s})^{2}+\mathcal{O}\left(s^{3}\right)+\mathcal{O}\left(s^{2}\ln^{3}{s}\right)+\mathcal{O}\left(s^{3}\ln{s}\right),\\
\kappa_{3}&=\frac{3}{2}s+\left[c_{30}+\left(\frac{117}{128}-\frac{3\gamma}{4}+\frac{3\gamma^{2}}{16}-\frac{\pi^{2}}{64}+\frac{3\ln{2}}{4}-\frac{3\gamma\ln{2}}{8}+\frac{3(\ln{2})^{2}}{16}\right)\ln{s}\right.\nonumber\\
&\quad\left.+\left(\frac{3\gamma}{32}-\frac{3}{16}-\frac{3\ln{2}}{32}\right)(\ln{s})^{2}+\frac{1}{64}(\ln{s})^{3}\right]s^{2}+\mathcal{O}\left(s^{3}\right)+\mathcal{O}\left(s^{2}\ln^{4}{s}\right)+\mathcal{O}\left(s^{3}\ln{s}\right),
\end{align*}
\noindent with $c_{30}$ given by
\begin{align*}
c_{30}=&-\frac{285}{256}+\frac{117\gamma}{64}-\frac{3\gamma^{2}}{4}+\frac{\gamma^{3}}{8}+\frac{\pi^{2}}{16}-\frac{\gamma\pi^{2}}{32}+\frac{3\gamma}{2}\ln{2}-\frac{3\gamma^{2}}{8}\ln{2}+\frac{\pi^{2}}{32}\ln{2}\nonumber\\
&-\frac{3}{4}\left(\ln{2}\right)^{2}+\frac{3\gamma}{8}\left(\ln{2}\right)^{2}-\frac{1}{8}\left(\ln{2}\right)^{3}+\frac{1}{32}\psi(2, 3)\approx 0.678.
\end{align*}
It is now clear that the moment generating function of the outage probability in small $s=4n^{2}/P$ (or equivalently, large transmitting power P) regime, the pdf of the outage probability
is not Gaussian. 
\par
For large $s,$ and from $(\ref{Q56}),$ we list below several cumulants,
\begin{align*}
k_{1}=&\frac{1}{2}\left(-1-2\alpha+\ln{2\pi}+2\alpha\psi(0, 1+\alpha)\right)+\sqrt{s}-\frac{\alpha}{2}\ln{s}+\frac{4-\alpha^{2}}{8}s^{-\frac{1}{2}}\nonumber\\
&+\left(\frac{3}{128}-\frac{5\alpha^{2}}{48}+\frac{\alpha^{4}}{24}\right)s^{-\frac{3}{2}}+\mathcal{O}\left(s^{-2}\right),\\
k_{2}=&-1+\psi(0, 1+\alpha)+\alpha\psi(1, 1+\alpha)-\frac{1}{2}\ln{s}+\frac{1-4\alpha^{2}}{8}s^{-1}+\mathcal{O}\left(s^{-2}\right),\\
k_{3}=&2\psi(1, 1+\alpha)+\alpha\psi(2, 1+\alpha)+\frac{1-4\alpha^{2}}{4}s^{-\frac{3}{2}}+\mathcal{O}\left(s^{-2}\right).
\end{align*}
For large $s,$ and $\alpha=0$ in (\ref{Q56}), we find,
\begin{align}\label{A8}
\ln\mathbf{\widetilde{M}}(s; 0, \lambda):
=&\lambda\left(\frac{1}{2}\ln\left(2\pi\right)-\frac{1}{2}+s^{\frac{1}{2}}+\frac{1}{8}s^{-\frac{1}{2}}+\frac{3}{128}s^{-\frac{3}{2}}+\frac{45}{1024}s^{-\frac{5}{2}}+\mathcal{O}\left(s^{-\frac{7}{2}}\right)\right)\nonumber\\
&+\frac{\lambda^{2}}{2!}\left(-1-\gamma-\frac{1}{2}\ln{s}+\frac{1}{8}s^{-1}+\frac{9}{64}s^{-2}+\frac{9}{16}s^{-3}+\mathcal{O}\left(s^{-4}\right)\right)\nonumber\\
&+\frac{\lambda^{3}}{3!}\left(2\zeta(2)+\frac{1}{4}s^{-\frac{3}{2}}+\frac{27}{32}s^{-\frac{5}{2}}+\mathcal{O}\left(s^{-\frac{7}{2}}\right)\right)\nonumber\\
&+\frac{\lambda^{4}}{4!}\left(-6\zeta(3)+\frac{3}{4}s^{-2}+\frac{45}{8}s^{-3}+\mathcal{O}\left(s^{-4}\right)\right)\nonumber\\
&+\sum_{j=5}^{\infty}\frac{\lambda^{j+1}}{(j+1)!}\left((-1)^{j}j!\,\zeta(j)+\mathcal{O}\left(s^{-\frac{5}{2}}\right)\right).
\end{align}
For comparison, we restate the cumulants obtained in Chen and McKay \cite{YangMcky2012}, where $P$ gets large, but in the variable $s=4n^2/P;$ see ((172),\cite{YangMcky2012}), and ((178),\cite{YangMcky2012}),
\begin{align*}
\kappa_{1}&\sim n\ln{P}-n +\frac{1}{8}\frac{1}{s^{1/2}}+\frac{3}{128}\frac{1}{s^{3/2}}+\frac{45}{1024}\frac{1}{s^{5/2}}+\cdots\\
\kappa_{2}&\sim \frac{1}{2}\ln{P}-\ln{2}+\frac{1}{8s}+\frac{9}{64\:s^2}+\cdots\\
\kappa_{3}&\sim \frac{1}{4n}+\frac{1}{4\:s^{3/2}}+\cdots.
\end{align*}
\par
If $k=0,$ or $\alpha=\frac{1}{2},$ then there is a rational solution
\begin{equation}\label{A10}
\mathcal{H}(s)=-\frac{1}{4}\lambda(1+\lambda)+\frac{\lambda\sqrt{s}}{2}, \qquad g(s)=\frac{\lambda}{\lambda-\sqrt{s}},
\end{equation}
which combines with (\ref{Q17}) to give
\begin{align}\label{A9}
\ln\mathbf{\widetilde{M}}\left(s; \frac{1}{2}, \lambda\right)&
:=\ln\left(\frac{G(1+\frac{1}{2}+\lambda)}{G(1+\frac{1}{2})}\right)+\lambda\left(\sqrt{s}-\frac{1}{4}\ln{s}\right)-\frac{\lambda^{2}}{4}\ln{s}\nonumber\\
&=\lambda\left(-1+\frac{1}{2}\ln(2\pi)+\frac{1}{2}\psi(0,\frac{3}{2})+\sqrt{s}-\frac{1}{4}\ln{s}\right)+\frac{\lambda^{2}}{2!}\left(-3+\frac{\pi^{2}}{4}+\psi(0,\frac{3}{2})-\frac{1}{4}\ln{s}\right)\nonumber\\
&\quad+\frac{\lambda^{3}}{3!}\left(-8+\pi^{2}+\psi(2,\frac{3}{2})\right)
+\frac{\lambda^{4}}{4!}\left(-48+\frac{\pi^{2}}{2}+3\psi(2,\frac{3}{2})\right)+\cdots,
\end{align}
which shows that the variance approximation breaks down if $s\geq 1.$ In fact, the fourth cumulant is \textit{negative}.
\par

\section{Limiting behavior of the reproducing kernel with the deformed Laguerre weight.}
In random matrix theory, Tracy and Widom \cite{TW1994} identified some reproducing kernels which arise from several matrix models, and which may be considered as exhibiting universal properties, independent of the particular model. The Bessel kernel is associated with `hard-edge' when one rescales the Laguerre ensemble at the edge of the spectrum. In this concluding section, we take the system of orthogonal polynomials for the deformed Laguerre weight and show they satisfy a system of ODE  in the style of  \cite{TW19942}. In Theorem \ref{TheoremA}, we show how this system behaves under double scaling, and then we recover the Bessel kernel in a particular case.

First, recall the lowering and raising operators in \cite{YangMcky2012}, and when we combine them with ((64), (65), (219) and (220), \cite{YangMcky2012})
we arrive at
\begin{equation}\label{Q41}
P'_{n}(x;t)+\left(\frac{r_{n}(t)}{x+t}-\frac{n+r_{n}(t)}{x}\right)P_{n}(x;t)=
\beta_{n}\left(\frac{1-R_{n}(t)}{x}+\frac{R_{n}(t)}{x+t}\right)P_{n-1}(x;t),
\end{equation}
\begin{equation}\label{Q42}
P_{n-1}'(x;t)+\left(\frac{\alpha+n+r_{n}}{x}+\frac{\lambda-r_{n}}{x+t}-1\right)P_{n-1}(x;t)=\left(\frac{R_{n-1}-1}{x}-\frac{R_{n-1}}{x+t}\right)P_{n}(x;t),
\end{equation}
where $R_{n}(t)$ and $r_{n}(t)$ given by ((221) and (222),  \cite{YangMcky2012}),
$\beta_{n}(t)=h_{n}(t)/h_{n-1}(t)$, and $h_{n}(t)$ is the square of the $L^{2}$ norm ((59),\cite{YangMcky2012}).
\par
With the aid of Christoffel--Darboux formula \cite{Szego1939}, one has the reproducing kernel,
\begin{equation}\label{Q43}
{\rm K}_{n}\left(x,y\right):=\sqrt{\frac{h_{n}(t)}{h_{n-1}(t)}}\frac{\varphi_{n}(x)\varphi_{n-1}(y)-\varphi_{n}(y)\varphi_{n-1}(x)}{x-y},
\end{equation}
where
\begin{equation*}
\varphi_{n}(x):=\frac{P_{n}(x;t)w^{\frac{1}{2}}(x;t)}{\sqrt{h_{n}(t)}}=\frac{P_{n}(x;t)}{\sqrt{h_{n}(t)}}
\left(x+t\right)^{\frac{\lambda}{2}}x^{\frac{\alpha}{2}}e^{-\frac{x}{2}},
\end{equation*}
\begin{equation*}
\varphi_{n-1}(x):=\frac{P_{n-1}(x;t)w^{\frac{1}{2}}(x;t)}{\sqrt{h_{n-1}(t)}}=\frac{P_{n-1}(x;t)}{\sqrt{h_{n-1}(t)}}
\left(x+t\right)^{\frac{\lambda}{2}}x^{\frac{\alpha}{2}}e^{-\frac{x}{2}}.
\end{equation*}
Combining these with the ladder operator relations of (\ref{Q41}), (\ref{Q42}) and the above definitions of $\varphi_{n}(x; t),$ $\varphi_{n-1}(x; t),$ we find
\begin{equation}\label{Q44}
\frac{d}{dx}\varphi_{n}(x;t)=\left(\frac{\alpha+2n+2r_{n}}{2x}+\frac{\lambda-2r_{n}}{2(x+t)}-\frac{1}{2}\right)\varphi_{n}(x;t)+\left(\frac{1-R_{n}}{x}+\frac{R_{n}}{x+t}\right)\beta_{n}^{\frac{1}{2}}\varphi_{n-1}(x;t),
\end{equation}
\begin{equation}\label{Q45}
\frac{d}{dx}\varphi_{n-1}(x;t)=-\left(\frac{1-R_{n-1}}{x}+\frac{R_{n-1}}{x+t}\right)\beta_{n}^{\frac{1}{2}}\varphi_{n}(x;t)-\left(\frac{\alpha+2n+2r_{n}}{2x}+\frac{\lambda-2r_{n}}{2(x+t)}-\frac{1}{2}\right)\varphi_{n-1}(x;t).
\end{equation}
Observe that we can write this as ${\frac{d}{dx}}\Phi_n(x)=\Omega_n(x;t)\Phi_n(x)$ where $\Omega_n(x;t)$ is a $2\times 2$ matrix with entries that are rational functions of $x$, and ${\rm{trace}}(\Omega_n(x,t)))=0$. This brings us within the context of \cite{TW19942}.
If we eliminate $\varphi_{n-1}(x; t)$ in (\ref{Q45}) with (\ref{Q44}), one find $\varphi_{n}(x;t)$ satisfies a second order ODE,
\begin{equation}\label{Q46}
D_{1}(x;t)\frac{d^{2}}{dx^{2}}\varphi_{n}+D_{2}(x;t)\frac{d}{dx}\varphi_{n}+D_{3}(x;t)\varphi_{n}=0,
\end{equation}
where
\begin{equation*}
D_{1}(x;t)=x+t-tR_{n},\qquad D_{2}(x;t)=1+\frac{t}{x}-\frac{tR_{n}}{x}-\frac{tR_{n}}{x+t},
\end{equation*}

\begin{align*}
D_{3}(x;t)=&-\frac{4H_{n}+\frac{\alpha^{2}}{2}+(1-R_{n})(\alpha{\lambda}+2H_{n}-2nt-\alpha{t})-(2n+\alpha+\lambda-t)R_{n}-2r_{n}+\lambda-t}{2x}\nonumber\\
&-\frac{\frac{\lambda^{2}}{2}+R_{n}(2H_{n}+\alpha{\lambda})+t\lambda{R_{n}}+(2n+\alpha+\lambda+t)R_{n}+2r_{n}-\lambda}{2(x+t)}+\frac{t\lambda^{2}R_{n}}{4(x+t)^{2}}\nonumber\\
&-\frac{t(1-R_{n})(\frac{\alpha^{2}}{4}+2H_{n})}{x^{2}}+\frac{2n+1+\alpha+\lambda}{2}-\frac{x+t(1-R_{n})}{4}.
\end{align*}
Hence the reproducing kernel of (\ref{Q43}) is characterized by the second order ODE, (\ref{Q46}).
\indent Let $f$ be a smooth real function of compact support, and for $\zeta\in \mathbb{R}$ consider the weight
\begin{equation}
w(x; t, \alpha,\lambda,\zeta )=(x+t)^{\lambda}x^{\alpha}e^{-x-\zeta f(x)},\;\;t>0,\;\;\alpha>-1,\; x\geq 0,
\end{equation}
which arises when we multiply $w(x;t,\alpha, \lambda ,\zeta)$ by $e^{-\zeta f(x)}$. We write $M_{e^{-\zeta f}-1}$ for the operation of multiplying by $e^{-\zeta f(x)}-1$. Corresponding to  (\ref{Q2}), we introduce
\begin{equation}\label{Q2}
D_{n}(t; \alpha, \lambda, \zeta )=\frac{1}{n!}\int_{{\mathbb{R}}_{+}^{n}}\prod_{j<k}(x_{j}-x_{k})^2\prod_{\ell=1}^{n}
\left(t+x_{\ell}\right)^{\lambda}x_{\ell}^{\alpha}
e^{-x_{\ell}-\zeta f(x_\ell)}dx_{\ell}.
\end{equation}
Then the ratio $D_{n}(t; \alpha, \lambda, \zeta )/D_{n}(t; \alpha, \lambda)$ is the moment generating function of $\sum_{\ell =1}^n f(x_\ell )$ under a suitable probability distribution. The kernel $K_n$ determines an integral operator on $L^2(0, \infty )$ so that 
\begin{equation}{\frac{D_{n}(t; \alpha, \lambda, \zeta )}{D_{n}(t; \alpha, \lambda)}}=\det \bigl( I+K_nM_{e^{-\zeta f}-1}),\end{equation} 
by standard results of random matrix theory. It is therefore of interest to know how $K_n$ behaves under the double scaling. 
 
\begin{theorem}\label{TheoremA}
Let
\begin{equation}\label{Q47}
X=4nx,\quad Y=4ny,\quad s=4nt,\quad \varphi(X; s):=\lim_{n\rightarrow\infty}\varphi_{n}\left(\frac{X}{4n}; \frac{s}{4n}\right).
\end{equation}
Let $n\rightarrow\infty,$ $t\rightarrow 0,$ $x\rightarrow 0,$ $y\rightarrow 0,$ such that $s,$ $X$ and $Y$ are in
compact subsets of $(0,\infty),$ 
then
\begin{align}\label{Q48}
\lim_{n\rightarrow \infty}\frac{1}{4n}{\rm K}_{n}\left(\frac{X}{4n},\frac{Y}{4n}\right)=\frac{A(Y)\varphi(X)-A(X)\varphi(Y)}{X-Y},
\end{align}
with  $A(z)$ given by
\begin{equation*}
A(z)=\frac{s}{2}\left(\frac{\lambda\left(\mathcal{H}(s)-\mathcal{H}'(s)\right)}{\mathcal{H}'(s)}
+4s\mathcal{H}''(s)\right)z\varphi(z)+\frac{z(z+s)}{z+s-4s\mathcal{H}'(s)}\varphi'(z).
\end{equation*}
Let
$\varphi'(z)$ denote $d\varphi(z)/dz$. Furthermore, $\varphi(z)$ satisfies the following second order ODE,
\begin{align}\label{Q49}
\varphi''&(z)+\left(\frac{1}{z}+\frac{1}{z+s}-\frac{1}{z+s-4s\mathcal{H}'(s)}\right)\varphi'(z)\nonumber\\
&-\left[\frac{\alpha^{2}}{4z^{2}}+\frac{\lambda^{2}}{4(z+s)^{2}}-\frac{z+s-2\alpha\lambda-4\mathcal{H}(s)}{4z(z+s)}
+\frac{\lambda{s}(\mathcal{H}'(s)-\mathcal{H}(s))
-4s^{2}\mathcal{H}'(s)\mathcal{H}''(s)}{2\mathcal{H}'(s)z(z+s)(z+s-4s\mathcal{H}'(s))}\right]\varphi(z)=0;
\end{align}
here $\mathcal{H}'$ denotes $d\mathcal{H}(s)/ds$ and $\mathcal{H}(s)$ is solution of the
$\sigma$-form of $P_V$ equation (\ref{Q14}).
\end{theorem}
\begin{proof}
Solving the equation (\ref{Q44}) in terms of $\varphi_{n}(x)$ as
\begin{equation}\label{Q50}
\varphi_{n-1}(x)=\frac{\left[-2n(t+x)+(x-\alpha)(t+x)-\lambda{x}-2tr_{n}\right]\varphi_{n}(x)+2x(t+x)\varphi_n'(x)}{2(x+t-tR_{n})\sqrt{\beta_{n}}}.
\end{equation}
Making use the above expression, we find,
\begin{align*}
R_{n}&=\frac{tH_{n}''(t)+\sqrt{[tH_{n}''(t)]^{2}+4[r_{n}^{2}(t)-\lambda{r_{n}(t)}]
[n(n+\alpha+\lambda)+tH'_{n}(t)-H_{n}(t)]}}{2[n(n+\alpha+\lambda)+tH'_{n}(t)-H_{n}(t)]}\nonumber\\
&\sim 4\mathcal{H}'(s)+\frac{\lambda\mathcal{H}(s)-4(\alpha+\lambda)(\mathcal{H}'(s))^{2}
+4s\mathcal{H}'(s)\mathcal{H}''(s)}{2n\mathcal{H}'(s)}.
\end{align*}
Proceeding further, we note that, (recall $4nt=s$),
\begin{equation*}
r_{n}=-H_{n}'(t)\sim -4n\mathcal{H}'(s),
\end{equation*}
obtained from \cite{YangMcky2012}. Substituting (\ref{Q47}) into the kernel (\ref{Q43}), we find that
(\ref{Q48}) follows immediately after some computations.
\par
If we eliminate $\varphi_{n-1}(x)$ in (\ref{Q45}) with the aid of (\ref{Q50}) and (\ref{Q47}), and
the estimation of  $R_{n}(t)$ and $r_{n}(t),$ then the second order ODE (\ref{Q49}) is obtained.
\end{proof}
\begin{rem12}
If $s=0,$ then the limiting kernel degenerates to the Bessel kernel obtained by \cite{TW1994} as,
\begin{equation}
\lim_{n\rightarrow \infty}\frac{1}{4n}{\rm K}_{n}\left(\frac{X}{4n},\frac{Y}{4n}\right)=\frac{Y\varphi'(Y)\varphi(X)-X\varphi'(X)\varphi(Y)}{X-Y},
\end{equation}
and the second order ODE (\ref{Q49}) becomes the Bessel differential equation,
\begin{equation}
\varphi''(X)+\frac{1}{X}\varphi'(X)+\left(\frac{1}{4X}-\frac{(\alpha+\lambda)^{2}}{4X^{2}}\right)\varphi(X)=0.
\end{equation}
Similarly, this degenerate case also appeared in the study of double scaling of the relevant kernel appearing
the singular weight problem in \cite{CCF2015}. Here $\varphi(X)$ is a constant multiple of $J_{\alpha +\lambda}(\sqrt {X}).$
\end{rem12}
\begin{rem13}
The weight (\ref{AA1}) can be compared with $x^{\alpha+\lambda}e^{-x+\frac{\lambda{t}}{x}},$ which is the so-called singularly perturbed Laguerre weight considered in \cite{ChenIts2010}. There the authors obtained a connection with $P_{III}$ for finite $n$.
Heuristically, we show this by recalling $s=4nt$, and taking the limit
\begin{align}
(t+x)^{\lambda}x^{\alpha}e^{-x}&=\Bigl({\frac{s}{4n}}+x\Bigr)^{\lambda}x^{\alpha}e^{-x}
=x^{\alpha+\lambda}\:\left(1+\frac{s}{4nx}\right)^{(4nx/s)\;\lambda\;t/x}\; e^{-x}\nonumber\\
&\to x^{\lambda+\alpha} e^{\lambda\;t/x-x},\qquad n\to\infty.\end{align}

See also \cite{XuDaiZhao2014} for the double scaling limit of the polynomial kernel, with physical background provided in \cite{OVAKE2007} and for a relevant study (for Hermite) see \cite{MFMO2009}. From the double scaling process, we indeed obtain a $P_{V}\left(\frac{\alpha^{2}}{2},-\frac{\lambda^{2}}{2},\frac{1}{2},0\right)$ (\ref{Q13}) which is equivalent to a particular $P_{III}.$
A referee suggested such a discussion from the view point of singularity of the weights, for which we would like to express our thanks.
\end{rem13}
\par
{\bf Acknowledgement.}
\par
We would like to thank the Science and Technology Development Fund of the Macau SAR for generous support:
FDCT 077/2012/A3 and FDCT 130/2014/A3. We also like to thank the University of Macau
for supporting us with MYRG 2014-00011-FST and MYRG 2014-00004-FST.
\\


\begin{thebibliography}{}
{\small \bibitem{BasorChen} {E. Basor} and {Y. Chen}, ``Perturbed Laguerre unitary ensembles, Hankel determinants, and information theory", {Math. Meth. Appl. Sci.} {38 (2015)}, {4840--4851}.
\bibitem{BP1998} {I. V. Barashenkov} and {D. E. Pelinovsky}, `` Exact vortex solutions of the complex sine-Gordon theory on the plane", {Phys. Lett. B} {436 (1998)}, {117--124}.
\bibitem{CC2015} {M. Chen} and {Y. Chen,} ``Singular linear statistics of the Laguerre unitary ensemble and Painlev\'{e}. III. Double scaling analysis", {J. Math. Phys.}, {56 (2015)}, {063506} (14pp).
\bibitem{CCF2015} {M. Chen,} {Y. Chen} and {E. G. Fan,} ``Perturbed Hankel determinant, correlation function and Painlev\'e equations", {J. Math. Phys.} {57 (2016)}, {023501} (31pp).
\bibitem{CFH2005} {H. Casini,} {C. D. Fosco} and {M. Huerta,} `` Entanglement and alpha entropies for
massive Dirac field in two dimensions", {J. Stat. Mech.:Theory Exp.} {P07007} ({2005}), (16pp).
\bibitem{YangHaqMcky2013} {Y. Chen,} {N. S. Haq} and {M. R. McKay,} ``Random matrix models, double-time Painlev\'{e} equations, and wireless relaying",
 {J. Math. Phys.} {54 (2013)}, {063506} (55pp).
\bibitem{ChenIsmail1997} {Y. Chen} and {M. E. H. Ismail,} ``Thermodynamic relations of the Hermitian matrix ensembles", {J. Phys. A.: Math. Gen.} {30 (1997)}, {6633--6654}.

\bibitem{ChenIts2010} {Y. Chen} and {A. Its,} ``Painlev\'{e} \uppercase\expandafter{\romannumeral3} and a singular linear statistics in Hermitian random matrix ensembles", {J. Approx. Theory} {162 (2010)}, {270--297}.

\bibitem{ChenLawrence1998} {Y. Chen} and {N. Lawrence,} ``On the linear statistics of Hermitian random matrices", {J. Phys. A.: Math. Gen.} {31 (1998)}, {1141--1152}.

\bibitem{YangMcky2012} {Y. Chen} and {M. R. McKay,} ``Coulomb fluid, Painlev\'{e} transcendents and the information theory of MIMO systems",
 {IEEE Trans. Inf. Theory} {58 (2012)}, {4594--4634}.
\bibitem{YM2010} {Y. Chen} and {M. R. McKay,} ``Perturbed Hankel determinants: Applications to the information theory of MIMO wireless communications'', arXiv:1007.0496 (2010), (77pp).
\bibitem{Clarkson2003} {P. A. Clarkson,} ``The third Painlev\'e equation and associated special polynomials", {J. Phys. A.} {36 (2003)}, {9507-9532}.
\bibitem{PAC2005} {P. A. Clarkson,} ``Special polynomials associated with rational solutions of the fifth Painlev\'e equation", {J. Comp. Appl. Math.} {178 (2005)}, {111-129}.
\bibitem{Dyson1962} {F. J. Dyson,} `` Statistical theory of energy levels of complex systems \uppercase\expandafter{\romannumeral1}-\uppercase\expandafter{\romannumeral3}", {J. Math. Phys.} {3 (1962)}, {140--175}.
\bibitem{FO2010} {P.J. Forrester} and {C.M. Ormerod}, ``Differential equations for deformed Laguerre polynomials'', {J. Approx. Theory} {162 (2010)}, {653--677}
\bibitem{FW} {P.J. Forrester} and {N.S. Witte,} ``Distribution of the first eigenvalue spacing at the hard edge of the Laguerre unitary ensemble} {Kyushu J. Math.} {61 (2007)}, {457--526} 
\bibitem{FoschiniGans1998} {G. J. Foschini} and {M. J. Gans,} ``On limits of wireless communications in a fading environment when using multiple antennas", {Wireless Pers. Commun.} {6 (1998)}, {311--335}.
\bibitem{PFW2006} {P. J. Forrester} and {N. S. Witte,} ``Boundary conditions associated with the Painlev\'e ${III}'$ and $V$ evaluations of some random matrix averages". {J. Phys. A: Math. Gen.} {39 (2006)}, {8983--8995}.
\bibitem{GLS2002} {V. I. Gromak,} {I. Laine} and {S. Shimomura,} ``Painlev\'e differential equations in the complex plane", {Vol. 28}, (Walter de Gruyter, 2002).
\bibitem{GR2007} {I. S. Gradshteyn} and {I. M. Ryzhik,} ``Table of integrals, series, and products", {7th ed.} (Elsevier/Academic Press, Amsterdam, {2007}).
\bibitem{HKL2008} {W. Hachem,} {O. Khorunzhiy,} {P. Loubaton,} {J. Najim} and {L. Pastur,} ``A new approach for mutual information analysis of large dimensional multi-antenna channels",
{IEEE Trans, Inf. Theory} {54 (2008)}, {3987-4004}.



\bibitem{H1996} {M. Hisakado,} ``Unitary matrix models and Painlev\'e {\rm III}", {Modern. Phys. Lett. A} {11 (1996)}, {3001-3010}.

\bibitem{JimboMiwa1981} {M. Jimbo} and {T. Miwa,} ``Monodromy perserving deformation of linear ordinary differential equations with rational coefficients, II", {Physica D} {2 (1981)}, {407-448}.

\bibitem{Jimbo1982} {M. Jimbo,} ``Monodromy problem and the boundary condition for some Painlev\'{e} equations", {Publ. RIMS, Kyoto Univ.} {18 (1982)}, {1137--1161}.

\bibitem{KM1999} {K. Kajiwara} and {T. Masuda,} ``On the Umemura polynomials for the Painlev\'e {\rm III} equation", {Phys. Lett. A} {260 (1999)}, {462--467}.

\bibitem{KMMC2011} {P. Kazakopoulos,} {P. Mertikopoulos,} {A. L. Moustakas} and {G. Caire,} ``Living at the edge: A large deviations approach to the outage MIMO capacity", {IEEE Trans. Inf. Theory} {57 (2011)}, {1984--2007}.
\bibitem{MC 2005} {M. R. McKay} and {I. B. Collings} ``General capacity bounds for spatially correlated Rician MIMO channels", {IEEE Trans. Inf. Theory} {51 (2005)}, {3121-3145}.

\bibitem{MFMO2009} {F. Mezzadri} and {M. Y. Mo,} ``On an average over the Gaussian Unitary ensemble", {Int. Math. Res. Not.} {2009 (2009)}, {3486-3515}.

\bibitem{MCB1995} {A. E. Milne,} {P. A. Clarkson} and {A. P. Bassom,} ``B\"acklund transformations and solution Hierarchies for the third Painlev\'e equation", {Stud. Appl. Math.} {98 (1997)}, {139--194}.
\bibitem{Murata1995} {Y. Murata,} ``Classical solutions of the third Painleve equation", {Nagoya Math. J.} {139 (1995)}, {37-65}.

\bibitem{N2004} {J. M. Normand,} ``Calculation of some determinants using the s-shifted factorial", {J. Phys. A.} {37 (2004)}, {5737--5762}.

\bibitem{Okamoto1981} {K. Okamoto,} ``On the $\tau$-function of the Painlev\'e equations", {Physica D} {2 (1981)}, {525-535}.

\bibitem{OLBC2010} {F. W. J. Olver,} {D. W. Lozier,} {R. F. Boisvert} and {C. W. Clark} (Editors), ``NIST Handbook of Mathematical Functions", (Cambridge University Press, 2010).


\bibitem{OVAKE2007} {V. A. Osipov} and {E. Kanzieper,} ``Are bosonic replicas faulty ?", {Phys. Rev. Lett.} {99 (2007)}, {050602}.

\bibitem{DMY2014} {D. Passemier,} {M. R. Mckay} and {Y. Chen,} ``Asymptotic linear spectral statistics for spiked Hermitian random matrix models", {J. Stat. Phys.} {160 (2015)}, {120-150}.

\bibitem{PS1990} {V. Periwal} and {D. Shevitz,} ``Unitary-matrix models as exactly solvable string theories", {Phys. Rev. Lett.} {64 (1990)}, {1326-1329}.

\bibitem{Szego1939}{G. Szeg{\"o},} ``Orthogonal polynomials", {American Mathematical Society}, {Vol. 23}, (American Mathematical Society Colloquium Publications, New York, {1939}).

\bibitem{SRS2003} {P. J. Smith,} {S. Roy} and {M. Shafi,} ``Capacity of MIMO systems with semicorrelated flat fading", {IEEE Trans. Inf. Theory} {49 (2003)}, {2781-2788}.

\bibitem{Telatar1999} {I. E. Telatar,} ``Capacity of  multi-antenna Gaussian channels", {Eur. Trans. Telecommun} {vol.10 (1999)}, {585--595}.

\bibitem{TW1994} {C. A. Tracy} and {H. Widom,} ``Level spacing distributions and the Bessel kernel", {Commun. Math. Phys.} {160 (1994)}, {289--309}.

\bibitem{TW19942} {C. A. Tracy} and {H. Widom,} ``Fredholm determinants, differential equations and matrix models", {Commun. Math. Phys.} {163 (1994)}, {33--72}.

\bibitem{TW1999} {C. A. Tracy} and {H. Widom,} ``Random unitary matrices, permutations and Painlev\'e", {Commun. Math. Phys.}, {Vol. 207 (1999)}, {665--685}.
\bibitem{TV2005} {A. M. Tulino} and {S. Verd\'{u},} ``Asymptotic outage capacity of multiantenna channels", {In: IEEE International Conference on Acoustics, Speech, and Signal Processing(ICASSP)} 5 (2005), {825--828}.
\bibitem{Voros} {A. Voros,} ``Spectral Functions, Special Functions, and the Selberg Zeta Function", {Commun. Math. Phys.} {110 (1987)}, {439-465}.
\bibitem{WW1958} {E. T. Whittaker} and {G. N. Watson,} ``A Course of modern analysis", (Cambridge: Cambridge University Press, {4th ed.} 1958).
\bibitem{XuDaiZhao2014} {S. X. Xu,} {D. Dai} and {Y. Q. Zhao,} ``Critical edge behavior and the Bessel to Airy transition in the singularly perturbed Laguerre Unitary ensemble", {Comm. Math. Phy.} {332 (2014)}, {1257-1296}.
\bibitem{ZhengTse2003} {L. Zheng} and {D. N. C. Tse,} ``Diversity and multiplexing: A fundamental tradeoff in multiple-antenna channels", {IEEE Trans. Inf. Theory} {49 (2003)}, {1073--1096}.
\end{thebibliography}
\end{document}